\documentclass[journal,10pt]{IEEEtran}
\usepackage{cite}
\usepackage[dvips]{graphicx}
\usepackage{graphicx}
\usepackage{xcolor}
\usepackage{amsmath}
\usepackage{amstext}

\usepackage{amsfonts}
\usepackage{amssymb}
\usepackage{psfrag}
\usepackage{fancyhdr}

\usepackage{subfigure}
\usepackage{epsfig}
\usepackage{csquotes}
\MakeOuterQuote{"}
\usepackage{enumerate}
\usepackage{bbm, dsfont}

\usepackage[normalem]{ulem} 
\usepackage{soul} 

\usepackage{color}

\usepackage[linesnumbered,ruled,vlined]{algorithm2e}

\usepackage{amsthm}

\newtheorem{thm}{Theorem}

\newtheorem{remark}{Remark}

\theoremstyle{definition}
\newtheorem{defn}{Definition}

\newtheorem{lem}{Lemma}
\newtheorem{cor}{Corollary}

\newtheorem*{exmp}{Example}

\usepackage{tabularx}
\usepackage{etoolbox}
\let\bbordermatrix\bordermatrix
\patchcmd{\bbordermatrix}{8.75}{4.75}{}{}
\patchcmd{\bbordermatrix}{\left(}{\left[}{}{}
\patchcmd{\bbordermatrix}{\right)}{\right]}{}{}
\usepackage{capt-of}

\begin{document}
\title{Dynamic Clustering and User Association in Wireless Small Cell Networks with Social Considerations}

\author{
\IEEEauthorblockN{Muhammad Ikram Ashraf\IEEEauthorrefmark{1}, Mehdi Bennis\IEEEauthorrefmark{1}\IEEEauthorrefmark{3}, Walid Saad\IEEEauthorrefmark{2}\IEEEauthorrefmark{3}, Marcos Katz\IEEEauthorrefmark{1} and Choong-Seong Hong\IEEEauthorrefmark{3}\\}
\IEEEauthorblockA{\small\IEEEauthorrefmark{1}Centre for Wireless Communications, University of Oulu, Finland, \\ email: \{ikram.ashraf,mehdi.bennis,marcos.katz\}@oulu.fi \\
\IEEEauthorrefmark{2}Wireless@VT, Bradley Department of Electrical and Computer Engineering, Virginia Tech, Blacksburg, USA, email: walids@vt.edu}\\
\IEEEauthorrefmark{3}{Department of Computer Engineering, Kyung Hee University, South Korea, email: cshong@khu.ac.kr}\vspace{-3ex}
\thanks{This research was supported by TEKES grant 2364/31/2014	and the Academy of Finland (CARMA) and
 the U.S. National Science Foundation under Grants CNS-1513697 and CNS-1460316.
}
}
\maketitle

\begin{abstract}
In this paper, a novel social network-aware user association in wireless small cell networks with underlaid device-to-device (D2D) communication is investigated. The proposed approach exploits social strategic relationships between user equipments (UEs) and their physical proximity to optimize the overall network performance. This problem is formulated as a matching game between UEs and their serving nodes (SNs) in which, an SN can be a small cell base station (SCBS) or an \emph{important UE} with D2D capabilities. The problem is cast as a many-to-one matching game in which UEs and SNs rank one another using preference relations that capture both the wireless aspects (i.e., received signal strength, traffic load, etc.) and users' social ties (e.g., UE proximity and social distance). Due to the combinatorial nature of the network-wide UE-SN matching, the problem is decomposed into a dynamic clustering problem in which SCBSs are grouped into disjoint clusters based on mutual interference. Subsequently, an UE-SN matching game is carried out per cluster. The game under consideration is shown to belong to a class of matching games with \emph{externalities} arising from interference and peer effects due to users social distance, enabling UEs and SNs to interact with one another until reaching a stable matching. Simulation results show that the proposed social-aware user association approach yields significant performance gains, reaching up to $26\%$, $24\%$, and $31\%$ for $5$-th, $50$-th and $95$-th percentiles for UE throughputs, respectively, as compared to the classical social-unaware baseline.
\end{abstract}
\vspace{-0.3cm}

\section{Introduction}
\label{sec:intro}
The proliferation of bandwidth intensive wireless applications such as multimedia streaming and online social networking (OSN) has led to a tremendous increase in wireless spectral resources \cite{Bennis2015}. This increasing need for wireless capacity mandates novel cellular architectures for delivering high quality-of-service (QoS) in a cost-effective manner. In this respect, small cell networks (SCNs), built on the premise of deploying inexpensive, low-power small cell base stations (SCBSs) are seen as a key technique to boost wireless capacity and offloading traffic. Reaping the benefits of SCNs requires overcoming a number of challenges that include user association, traffic offloading, resource management, among others \cite{Bennis2015, JA2013, DLopez2011, GdeRoche2010}. Along with the rapid proliferation of SCNs, cellular systems are moving from a base station to a user-centric architecture driven by the surge of user specific applications \cite{Bennis2013}. It is anticipated that a large number of devices with varying QoS requirements will interact within small coverage footprints \cite{Federico2014}. Hence, in conjunction with SCNs, device-to-device (D2D) communication over cellular bands has emerged as a promising technique to further improve the performance of SCNs, in which D2D devices communicate directly bypassing the infrastructure yielding increased network capacity, extended coverage, enhanced data offload and improved energy efficiency \cite{Federico2014,AQD2d2014, GED2d2012, LSongBook2014, YPei2013, JinLi2013}. The 3GPP LTE Release 12 has dealt with D2D communication in order to address the ever-increasing demands for data traffic.

The benefits of D2D communication are accompanied with a number of technical challenges that include proximity service discovery (ProSe), resource allocation, and intercell interference coordination between cellular and D2D links \cite{AQD2d2014, GED2d2012, LSongBook2014}. In particular, one key challenge in D2D-enabled SCN is that of associating user equipments (UEs) to their preferred serving node (SN) that can be either a SCBS or other D2D users. In \cite{YPei2013}, the authors present a protocol for resource allocation and selection of potential D2D SNs to improve the sum rate of D2D links. In \cite{JinLi2013}, an optimization problem is formulated enabling D2D links to improve their resource utilization and aggregate link capacity. Most of the existing works on SCNs and D2D enabled user association are focused on conventional physical layer metrics to optimize the network performance \cite{Bennis2015, JA2013, DLopez2011, GdeRoche2010, AQD2d2014, GED2d2012, JinLi2013}. To this end, one promising approach for addressing the user association problem is to incorporate additional contextual information such as user's social ties, network connectivity and other features to further boost the network performance. For example, in a football stadium, a group of neighboring friends may like to share the statistics of a player. Coupled with their physical proximity, the social networking relationships between these users can provide an indication on their common interests to share the same content. In a conventional setting, a SCBS often ends up serving different users with the same content using multiple duplicate transmissions which leads to a waste of resources and degrades the overall QoS. Social network-aware user association, as presented in this work, is a new paradigm to boost the performance of SCNs by exploiting D2D communications.

However, incorporating different contextual information in conjunction with conventional physical layer metrics enables better resource utilization and enhanced traffic offloading \cite{Dimitrio2014, Proebster2011, Magnus2012}. In \cite{Dimitrio2014}, the authors presented a radio resource management technique which incorporates multiple context information (spectrum bands, QoS, location) within the SCNs which leads to better spectrum usage. A network utility maximization problem is solved by exploring contextual information at UEs such as application's foreground/backgroud state in order to improve QoS \cite{Proebster2011}. A scheduling algorithm is developed for cellular wireless networks, which utilizes the information captured from users's environment (packet flow and delay requirement) to examine the throughput-delay tradeoff \cite{Magnus2012}. In this respect, the authors in \cite{LChen2013} propose a self-organizing cluster-based load balancing scheme for traffic offloading while, the authors in \cite{OAnjum2014} propose a decentralized coordination mechanism with a focus on cell edge users based on system level simulations. However, while interesting, these works are limited to conventional wireless systems relying on a central controller which can cause significant information exchange, and thus will not be appropriate for dense SCNs. This motivates the need for decentralized and self-organizing resource management solutions.

The main contribution of this work is to propose a novel, dynamic clustering and social-aware user association mechanism in D2D-enabled SCNs. Unlike previous works \cite{LChen2013, Abelnaseer2014, KHosseini2012}, we propose a clustering approach that incorporates both location and traffic load of SCBS and specifically, incorporates conventional channel information and social interaction between users to optimize user association in D2D enabled SCNs. In order to exploit the social-ties among nodes, we utilize the notion of \emph{social distance} to identify sets of socially important nodes acting as the best SNs for other UEs within proximity range. In the proposed model, the decision of UEs on whether to use a cellular or D2D link takes into account the social importance of the node in conjunction with the traffic load, channel conditions and interference. We formulate the problem as a many-to-one matching game per cluster with \emph{externalities} in which the serving nodes (i.e., SCBS and/or important UE) and UEs are the players, which rank one another based on set of preferences seeking suitable and stable association. The use of coalition formation games for D2D scenarios, as studied in the literature, typically seeks to maximize resource utilization and enhance network performance such as in \cite{YXiao2015} and \cite{YShen2016}. In \cite{YXiao2015}, the authors presented spectrum sharing problem as a Bayesian non-transferable utility overlapping coalition formation (BOCF) game between a set of device-to-device (D2D) links and multiple co-located networks. In \cite{YShen2016}, the authors studied D2D coalition formation among UEs in a single cell for video sharing scenario with peak signal-to-noise ratio (PSNR) as the measurement for video quality. Unlike \cite{YXiao2015} and \cite{YShen2016}, here, we consider an equally loaded multi-cell system such that spectrum is shared in the first time slot, whereas full spectrum is utilized in the subsequent time slot for D2D-transmissions. Furthermore, physically shared links with interference and social interactions among UEs are incorporated into the proposed matching game.

Many works have been presented in the literature to solve numerous matching markets in microeconomics such as \cite{ARoth1992, DFMan2013, Bando2012}. Unlike previous works \cite{ARoth1992, DFMan2013, Bando2012, EA2011, SBayat2012, ALeshem2011}, the strategy of each player in the proposed matching game is affected by the decisions of its peers. In this regard, the works in \cite{YHYang2013} and \cite{CJiang2015} deal with network externalities, however, they consider different types of optimization/game-theoretic problems and do not focus on matching games with externalities as studied here. In particular, we show that the proposed game belongs to a class of matching games with externalities (i.e., negative externalities) arising from interference and peer effects (i.e., social interaction) between nodes, which distinctly differs from the prior works presented in \cite{EA2011, SBayat2012, ALeshem2011, YHYang2013, CJiang2015}. To solve this game, we propose a distributed algorithm that allows UEs and SNs to self-organize and to maximize their own utilities within their respective clusters. In addition, the proposed algorithm is shown to converge to a stable matching in which no player has an incentive to match to other player, even in the presence of externalities. We use the concept of \emph{two-sided pairwise  matching} to prove stability, in which UEs and SNs can swap their association preference in order to maximize their utility. Simulation results validate the effectiveness of the proposed approach, and show significant performance gains compared to the baseline social-unaware user association approaches.

The rest of this paper is organized as follows. In Section \ref{sec:2}, we present the system model followed by the wireless and social network models. The social network-aware user association problem is formulated in Section \ref{sec:3}. The dynamic clustering and intra-cluster coordination is detailed in Section \ref{sec:4}. In Section \ref{sec:5}, we study the UE-SN association as a matching game with externalities and discussed its properties. Simulation results are presented in Section \ref{sec:6}. Finally, conclusions are drawn in Section \ref{sec:7}.
\section{System Model}
\label{sec:2}
\subsection{Wireless Network Model}
\label{sec:21}
Consider the \textit{downlink} transmission of a macrocell network underlaid by $N$ SCBSs. We assume that all SCBSs transmit on the same frequency spectrum (i.e., co-channel deployment) with bandwidth $B$. Let $\mathcal{N} = \{1, \ldots, N\}$, $\mathcal{M} = \{1, \ldots, M\}$, and $\mathcal{I} = \{1, \ldots, I\}$, $\mathcal{I} \subset \mathcal{M},\; \mathcal{I} \neq \mathcal{M}$, be the sets of SCBSs, UEs, and important UEs, respectively. Here, an important UE is defined as a socially well connected node within a confined coverage area serving as anchor node\footnote{The terms Important UE and anchor node are used interchangeably.} for D2D communication. We let $\mathcal{P} = \{1, \ldots, P\}$ be the set of SNs, which can be either SCBSs or important UEs, i.e., $\mathcal{P} = \mathcal{N} \cup \mathcal{I}$. We let $\mathcal{L}_{n}$ be the set of UEs serviced by SCBS $n$ and $\mathcal{M}_i$ be the set of UEs serviced by the important node $i \in \mathcal{I}$. Let $\mathcal{M}_u$ be the set of $M_u$ non-serving UEs such that, $\mathcal{M} = \mathcal{M}_u \cup \mathcal{I}$. All the symbols which are used in the rest of the paper are summarized in Table \ref{tab:symbols}. The considered network model is shown in Fig. \ref{fig:physical-model}. We assume only slowly-varying channel state information (CSI) at the SCBS \cite{VHa2014}. Moreover, in our model users are not capable of transmitting and receiving simultaneously, so half duplex UEs are considered. In the first time slot $\tau_0$, SCBS $n$ transmits to UE $\hat{m} \in \mathcal{M} \setminus \mathcal{M}_i, \forall i \in \mathcal{I}$ while in the next time slot $\tau_1$, important UE $i$ decodes and forwards its received signal to UE $m \in \mathcal{M}_i$. Thus, the achievable rate between SCBS $n \in \mathcal{N}$ and UE $\hat{m} \in \mathcal{M} \setminus \mathcal{M}_i, \forall i \in \mathcal{I}$ at time slot $\tau_0$ is given by:
\begin{table}[t]
	\footnotesize
	\caption{Summary of Important Symbols}
	\centering
	\begin{tabular}{| p{1.0cm} | p{6.5cm} |}
		\hline
		\textbf{Symbol} & \textbf{Description} \\ \hline
		$\mathcal{N}$     & Set of SCBSs in the network    \\ 
		$\mathcal{M}$     & Set of UEs in the network \\ 
		$\mathcal{I}$     & Set of important UEs  \\
		$\mathcal{M}_i$   & Set of UEs serviced by important UE $i$ \\
		$\mathcal{P}$     & Set of SNs     \\
		$\mathcal{C}$     & Set of Clusters \\
		$B$     & Bandwidth                  \\ 
		$\mathcal{L}_n$   & Set of UEs serviced by SCBS $n$ \\
		$R_{n,m}$ & Achievable rate between SCBS $n$ and UE $m$ \\
		$\widetilde{R}_{n,m}$ & Achievable rate between SCBS $n$ and D2D UE $m \in \mathcal{M}_i$\\
		$p_{n}$           & Transmit power of SCBS $n$     \\
		$h_{n,m}$ & Channel gain between SCBS $n$ and UE $m$ \\
		$\boldsymbol{S}$  & Similarity matrix  \\             
		$\boldsymbol{A}$   & Edge betweenness centrality matrix for UEs \\
		$\boldsymbol{X}$  & Weighted cost matrix for UEs \\
		$\boldsymbol{W}$   & Social distance matrix for UEs \\
		$w_{m,\tilde{m}}$ & Social distance between UEs $m$ and $\tilde{m}$ \\
		$U_{p,m}$         & Utility of UE $m$ with respect to SN $p$ \\
		$\rho_n$   & Total load of SCBS $n \in \mathcal{N}$ \\
		$U_p$             & Utility of SN $p$ \\
		$\Gamma(\eta_c)$ & Social welfare of cluster $c \in \mathcal{C}$ for given matching $\eta_c$ \\
		$s^d_{n_1,n_2}$   & Gaussian distance similarity between SCBS $n_1, n_2 \in \mathcal{N}$ \\ $s^l_{n_1,n_2}$  & Gaussian load dissimilarity between SCBS $n_1,n_2 \in \mathcal{N}$ \\
		$\boldsymbol{D}$  & Gaussian distance similarity matrix \\
		$\boldsymbol{L}$            & Gaussian load dissimilarity matrix \\
		$\boldsymbol{Y}$  & Gaussian affinity matrix \\           
		$\boldsymbol{H}$            & Degree matrix  \\
		$\Omega$          & Tunable parameter control the impact of distance and load on similarities \\
		$\succ $ & preference relation \\
		\hline
	\end{tabular}
	\label{tab:symbols}
\end{table}
\begin{equation}
\label{eq:rate-scbs-cu}
R_{n,\hat{m}}=\frac{\tau_0}{T} \cdot \frac{B}{|\mathcal{L}_n|}\cdot\log_2\left(1+\frac{p_{n} h_{n,\hat{m}}}{N_0 +\sum_{n^{\prime}\in\mathcal{N}\setminus\{n\}}p_{n^{\prime}} h_{n^{\prime},\hat{m}}}\right),
\end{equation}
\noindent
where $T$ is the time duration for a frame such that $T = \tau_0 + \tau_1$, $p_{n}$ is the transmission power of SCBS $n$, $h_{n,\hat{m}}$ is the channel gain from SCBS $n$ to UE $\hat{m}$, respectively, while $N_0$ is the noise spectral density. The interference term in the denominator represents the aggregate interference at UE $\hat{m}$ caused by the transmissions of other SCBSs $n^{\prime}\in\mathcal{N}\setminus\{n\}$. We assume that the important UE sends (the same) content to all D2D UEs within the cell. Therefore, the rate between important UE $i \in \mathcal{I}$ and UE $m \in \mathcal{M}_i$ at time slot $\tau_1$ is:
\begin{equation}
\label{eq:rate-scbs-imp}
    R_{i,m}= \min_{\forall m \in \mathcal{M}_i} \bigg[ \frac{\tau_1}{T} \cdot B \cdot\log_2 \bigg(1+\frac{p_{i} h_{i,m}}{N_0+\sum_{i^{\prime}\in\mathcal{I}\setminus\{i\}}p_{i^{\prime}} h_{i^{\prime},m}}\bigg) \bigg],
\end{equation}
\noindent
where $p_{i}$ is the transmission power of important UE $i$ and  $h_{i,m}$ is the channel gain from the important UE $i$ to a given UE $m$, respectively. The interference term in the denominator represents the aggregate interference at UE $m$ caused by the transmissions of other important UEs $i^{\prime}\in\mathcal{I}\setminus\{i\}$. The achievable rate between SCBS $n$ and D2D UE $m \in \mathcal{M}_i$ over $T = \tau_0 + \tau_1$ is:
\begin{equation}
\label{eq:rate-d2d}
   \widetilde{R}_{n,m} = \min(R_{n,i}, R_{i,m}).
\end{equation}
\begin{figure}[t]
	\centering
	\includegraphics[width = \columnwidth]{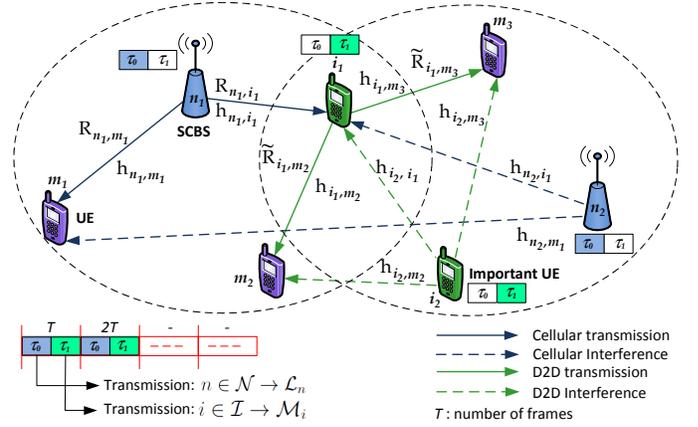}
	\caption{Network model for D2D-enabled SCNs.}
	\label{fig:physical-model}
\end{figure}
Our objective is to propose an efficient and self-organizing user association scheme for D2D-enabled SCNs. In conventional SCNs, each UE is associated to a SCBS based on the maximum signal-to-interference-plus-noise (SINR) or highest received signal strength indicator (RSSI) \cite{Bennis2015} ignoring UEs' contextual information such as proximity services, network and social ties. This motivates for investigating novel social-aware user association mechanism in D2D enabled SCNs.
\subsection{Social Network Model for UEs}
\label{sec:22}
For a more efficient user association, we define the notion of \emph{social tie} which characterizes the strength of the social relationship between two nodes. Here, we assume that a D2D link between two nodes is formed if they are socially connected and they are within proximity range. In order to establish D2D links, some UEs can be selected as important nodes to serve other UEs within proximity. For instance, in the context of content sharing leveraging users social ties allows the SCBS to avoid sending multiple copies of the same content. Instead, by selecting socially important nodes as caching points, UEs communicate via D2D links within the same social network, thereby offloading the base station. In particular, the social network can be represented by a weighted graph whose vertices represent nodes and edges represent their relationships strength based on parameters such as friendship or common interests. We use the concept of \emph{social distance} to measure the strength of a link between two nodes. Let $\mathcal{G}^s = (\mathcal{M},\mathcal{E}, w)$ be the social graph, where $\mathcal{M}$ is the set of UEs, and $\mathcal{E}$ is the set of edges. The social distance $w_{m,\tilde{m}}$ is the weight of the edge $e\in \mathcal{E}$ between UEs $m$ and $\tilde{m}$, and adjacent UEs $(m,\tilde{m})$ are connected via an edge $e$. Moreover, the social distance matrix $\boldsymbol{W}$ is symmetric such that $w_{m,\tilde{m}} = w_{\tilde{m},m}, \forall m, \tilde{m}$.

\subsubsection{Important UE} An important UE is a UE that is socially popular or well-connected as compared to other UEs in the network. Popularity or centrality in social network graphs quantify the importance of a peak in such graphs, or popularity of a node in social networks. Evidently, a node with high popularity has high probability of having a link to other network nodes. Hence, the social importance can be characterized by having curtail points for data distribution in the network, since it has social ties/links with other nodes in the network. The three most popular ways to quantify the social popularity of nodes in a social graph are degree, closeness, and betweenness centrality \cite{Ulrik2008, PC2006}. In this work, the \emph{social importance} of a UE is defined as a mixture of edge betweenness centrality, similarity, and physical distance to other peer nodes. The concept of \emph{social distance} between nodes is based on edge betweenness and node similarity.

\subsubsection{Social Distance} The social distance is defined as the social interaction parameter between communicating nodes. \textit{Important UEs} in a given cell can be interpreted as the subset of UEs with the highest social distance for data transfer. Let $\boldsymbol{W}$ be a social distance matrix where element $w_{m,\tilde{m}}\in[0,1]$ quantifies how the social distance of a user affects its utility which is given as the weighted sum of matrices $\boldsymbol{S}$, $\boldsymbol{A}$ representing respectively the similarity and edge betweenness centrality among UEs \cite{Eli2007}:
\begin{equation}
	\label{eq:socialdistance}
	\boldsymbol{W} = \alpha \boldsymbol{S} + \beta \boldsymbol{A},
\end{equation}
\noindent
where $\alpha$ and $\beta$ given in (\ref{eq:socialdistance}) are tunable parameters such that $\alpha + \beta = 1$. The \textit{important UE} selects its preferred peer based on the composite social and physical distance captured by the following weighted cost matrix $\boldsymbol{X}$, where element $x_{m,\tilde{m}}$ is given by:
\begin{equation}
	\label{eq:weightcost}
        x_{m,\tilde{m}} = (\epsilon_{m,\tilde{m}} {w}_{m,\tilde{m}})/d_{m,\tilde{m}}.
\end{equation}
In (\ref{eq:weightcost}), we combine the social distance with the actual physical distance between UEs, where $w_{m,\tilde{m}}$ denote the social distance, $d_{m,\tilde{m}}$ is the Euclidean physical distance between UE $m$ and $\tilde{m}$ and $\epsilon_{m,\tilde{m}}$ is a normalization constant. In order to establish a D2D link between UEs, we assume that some UEs are selected as important UEs. The social distance is used for ranking the nodes for selecting popular (important) nodes in the social network. A UE is considered important if its aggregate weighted cost is larger than other UEs in the proximity of SCBS $n$, such that $\mathcal{I}_n = \text{argmax}_{\forall \tilde{m} \in \mathcal{M}} \sum_{m \in \mathcal{M}, m \neq \tilde{m}}x_{m \tilde{m}}$. We use an adjacency matrix $\boldsymbol{E}$ to determine the existence of a D2D link, where element $e_{m,i}=1$, if $m$ is connected to important UE $i$, otherwise $e_{m,i}=0$. Next, we briefly review the concepts of similarity and betweenness centrality to capture the social distance among UEs in a given social networks.

\subsubsection{Similarity Matrix} The similarity matrix is a measure of closeness between a pair of nodes. The degree of similarity can be measured by the ratio of common neighbors between individuals in a social network. The degree of similarity between UEs $m$ and $\tilde{m}$ has an important effect in terms of data dissemination. Nodes having lower degree of similarity are good candidates for data dissemination \cite{TZhou2009}. Let $\boldsymbol{Q}$ be a $M \times M$ similarity matrix, such that a pair of nodes $(m,\tilde{m})$, depending on whether they are connected directly or indirectly, their corresponding similarity measuring element $q_{m,\tilde{m}}$ of $\boldsymbol{Q}$ is defined as \cite{TZhou2009}:
\begin{equation}
\label{eq:similarity}
			q_{m,\tilde{m}} = \begin{cases}	\displaystyle \sum_{\widehat{m} \in \nu(m)\cap\nu(\tilde{m})}\frac{1}{t(\widehat{m})}, & \mbox{if $m,\tilde{m}$ are connected},\\
        0, & \mbox{otherwise},\end{cases}
\end{equation}
\noindent
where $\nu(m)$ is the set of neighbors of $m$, $\;\widehat{m} \in \nu(m) \cap \nu (\tilde{m})$ are the common neighbors of UEs $m$ and $\tilde{m}$, and $t(\widehat{m})$ is the degree of UE $\widehat{m}$. To normalize the similarity matrix, we use the simple additive weighting (SAW) method, in which the normalized value of each element $q_{m,\tilde{m}}$ of $\boldsymbol{Q}$ is:
\begin{equation}
\label{eq:normsim}
	s_{m,\tilde{m}} = q_{m,\tilde{m}}/q_{\tilde{m}}^{\max} \;\; \forall m,\tilde{m},
\end{equation}
where $q_{\tilde{m}}^{\max}=\max_{m} q_{m,\tilde{m}}$. Consequently, we obtain the normalized similarity matrix $\boldsymbol{S}$ of dimension $M \times M$, where the $m$th row and $\tilde{m}$th column of $\boldsymbol{S}$, i.e., $s_{m,\tilde{m}}$ denotes the normalized similarity between UEs $m$ and $\tilde{m}$.

\subsubsection{Edge Betweenness Centrality} Edge betweenness centrality is based on the idea that an edge becomes central to a graph if it lies between many other UEs, i.e., it is traversed by many of the shortest paths connecting a pair of UEs \cite{PC2006}. Edges with a high betweenness centrality are considered important because they control information flow in the social network. Let $\boldsymbol{A}$ be $M \times M$ edge betweenness centrality matrix, where element $a_{m,\tilde{m}}$ is the edge betweenness centrality of the link between nodes $m$ and $\tilde{m}$. The betweenness centrality $a_{m,\tilde{m}}$ of an edge $e$ \cite{Ulrik2008} between UEs $(m,\tilde{m})$ is the sum of the fraction of all-pairs' shortest paths that pass through edge $e$. The normalized $a_{m,\tilde{m}}$ is:
\begin{equation}
\label{eq:betweenness}
	 a_{m,\tilde{m}} = \frac{ \displaystyle \sum_{m,\tilde{m} \in \mathcal{M}} \frac{\gamma (m,\tilde{m}|e)}{\gamma(m,\tilde{m})} } {(M-1) (M-1)},
\end{equation}
\noindent
where $M$ is the number of UEs, the summation $\gamma(m,\tilde{m})$ is over the number of shortest $(m,\tilde{m})$-paths, and $\gamma (m,\tilde{m}|e)$ is the number of those paths that traverse edge $e$. To provide more insights on this social model, we present the following example to compute the social distance (\ref{eq:socialdistance}) and weighted cost (\ref{eq:weightcost}).

\begin{exmp}
Consider four UEs from the set $\mathcal{M} = \{m_{1},m_{2},m_{3},m_{4}\}$ and one SCBS, $ n_1 \in \mathcal{N}$ as shown in Fig. \ref{fig:exampleNetwork}. Formally, the connectivity of UE $m\in \mathcal{M}$ can be represented by an adjacency matrix $\boldsymbol{E}$, which is a $M \times M$ symmetric matrix, where $M$ is the number of UEs in the social graph $\mathcal{G}^s$. The adjacency matrix is expressed as:
\begin{figure}[t]
\centering
\includegraphics[width = \columnwidth]{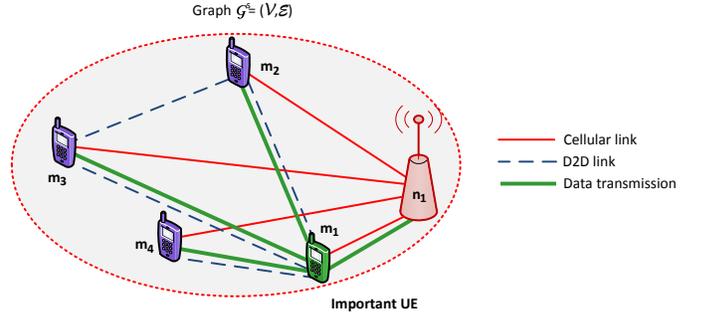}
\caption{Illustrative example of the considered network deployment with one SCBS $(n_1)$, one important UE $(m_1)$ and three UEs $(m_2, m_3, m_4)$ represented as graph $\mathcal{G}^s=(\mathcal{V}, \mathcal{E})$.}
\label{fig:exampleNetwork}
\end{figure}%
\begin{equation}
			E_{m,\tilde{m}} = \begin{cases}	1, & \mbox{if there is a edge between UEs $m$ and $\tilde{m}$},\\
        0, & \mbox{otherwise}.\end{cases}
				\nonumber
\end{equation}
The similarity $\boldsymbol{S}$ and edge betweenness centrality $\boldsymbol{A}$ between UEs computed using (\ref{eq:similarity}) and (\ref{eq:betweenness}), respectively.
\vspace{-0.3cm}
\begin{equation}
\nonumber
\boldsymbol{S} = \bbordermatrix{~ &  &  &  &    \cr
     	m_{1} & 0     & 0.5833 & 0.583 & 0.2500 \cr
		m_{2} & 0.583 & 0      & 0.500 & 0.500  \cr
		m_{3} & 0.583 & 0.500  & 0     & 0.500  \cr
        m_{4} & 0.250 & 0.500  & 0.500 & 0      \cr}
\end{equation}
\vspace{-0.5cm}
\begin{equation}
\nonumber
\boldsymbol{A} = \bbordermatrix{~ &  &  &  & \cr
		m_{1} & 0      & 0.0750 & 0.0750 & 0.100 \cr
		m_{2} & 0.0750 & 0      & 0.0500 & 0      \cr
		m_{3} & 0.0750 & 0.0500 & 0      & 0      \cr
        m_{4} & 0.1000 & 0      & 0      & 0      \cr}
\end{equation}
\vspace{-0.5cm}
\begin{equation}
\nonumber
\boldsymbol{W} = \bbordermatrix{~ &  &  &  & \cr
  m_{1} & 0      & 0.3292 & 0.3292 & 0.1750 \cr
  m_{2} & 0.3292 & 0      & 0.2750 & 0.2500 \cr
  m_{3} & 0.3292 & 0.2750 & 0      & 0.2500 \cr
  m_{4} & 0.1750 & 0.2500 & 0.2500 & 0      \cr}
\end{equation}
\vspace{-0.5cm}
\begin{equation}
\nonumber
\boldsymbol{X} = \bbordermatrix{~ &  &  &  & \cr
  m_{1} & 0      & 0.0432 & 0.0411   & 0.01240 \cr
  m_{2} & 0.0432 & 0      & 0.0320   & 0.0117 \cr
  m_{3} & 0.0411 & 0.0320 & 0        & 0.0121 \cr
  m_{4} & 0.0124 & 0.0117 & 0.0121   & 0      \cr}
\end{equation}
To determine which UE is the most important UE with respect to SCBS $n_1$, we use (\ref{eq:socialdistance}) with $\alpha = 0.5 $ and $\beta = 0.5$. By looking at the social distance matrix $\boldsymbol{W}$ and calculating the respective weighted cost matrix (\ref{eq:weightcost}), it is clear that UE $m_{1}$ is more socially important than other UEs while, $m_{4}$ is the least important one.
\end{exmp}
\section{Problem Formulation}
\label{sec:3}
As previously mentioned, classical approaches for user association in SCNs, are typically based on physical layer metrics and assume a central controller which gathers all network information and decisions \cite{Bennis2015}. In this section, we study the problem of base station clustering and flexible user association by incorporating users' social-ties in the network. Then, we will use the framework of matching theory \cite{Eliz2011}, to develop a distributed and self-organizing solution composed of two steps: 1) we cluster SCBSs in terms of mutual interference described in detail in Section \ref{sec:4}, 2) we study a two-sided matching model that enables each cluster to efficiently optimize user association by incorporating both physical and social aspects. Therefore, we define a two-sided matching game in which UE and SN acts as players. In this game each player tries to match (associate) to the most suitable serving node based on its own preference $\eta: \mathcal{M} \rightarrow \mathcal{P}$. Next, we define the social-aware utility functions which capture both wireless and social network metrics in order to optimize the user association mechanism.
\subsection{UE and SN Utilities}
The utility of a given UE is defined as the achievable rate taking into account the interference from adjacent SCBSs and important UEs. An arbitrary UE $m$ can either connect to a SCBS $n$ via a cellular connection or an \textit{important UE} $i \in \mathcal{I}$, via a \textit{D2D link}. The achievable rate between SN $p \in \mathcal{P}$ (important UE or SCBS) and UE $m \in \mathcal{M}$ for a given matching $\eta$ is:
\begin{align}
\label{eq:utility-ue}
   & U_{p,m} (R_{p,m},w_{p,m}, \eta)  \nonumber \\
     & = \begin{cases} R_{p,m} + \displaystyle \sum_{\tilde{m} \in \mathcal{M} \setminus{m}} \frac{\widetilde{R}_{m,\tilde{m}}}{1-w_{m,\tilde{m}}}e_{m,\tilde{m}}, \;\; \mbox{($p=n$ and $m=i$)},\\
   R_{p,m}, \;\; \mbox{if $m$ connected to SN $p$ ($p=n$)}, \\
   \widetilde{R}_{p,m}, \;\; 
   \mbox{D2D UE $m$ as per (\ref{eq:rate-d2d})},
   \end{cases}
\end{align}
\noindent
where $w_{p,m}$ represents the social distance between SN $p$ and UE $m$ in the social graph $\mathcal{G}^s$ defined in (\ref{eq:weightcost}). The element $e_{m,\tilde{m}} \in \{0,1\}$, shows the existence of a D2D-link between UE $m$ and UE $\tilde{m}$. Moreover, \eqref{eq:utility-ue} defines the utility of a UE when it acts as an important UE $m=i$ serviced by SCBS $p=n$ and forwards data to other UEs within a given social network of UEs. Therefore, in order to capture the social impact between a pair of UEs $m$ and $\tilde{m}$ that have social ties between them in the social network, we formalize the strength of the social tie (social distance) as $w_{m, \tilde{m}} = [0,1)$, with a higher value of $w_{m, \tilde{m}}$ being a stronger social tie. It follows that the social utility of an important UE consists of its achievable rate and a weighted sum of the achievable rates of other UEs having social tie with it \cite{XChen2014}. This will induce socially well connected UEs to associate to one another.

To calculate the utility of SN $p \in \mathcal{P}$, we incorporate the social distance of each UE $m$ with respect to SN $p$ \cite{Strat2011}. The utility of an SN $p$ is the sum of utilities of its associated UEs $m \in \mathcal{L}_{p}$, for a matching $\eta$ given by:
\begin{align}
\label{eq:utility-scbs}
	U_{p}(\eta) &=\sum_{m \in \mathcal{L}_{p}}U_{p,m}(R_{p,m},w_{p,m}).
\end{align}
\subsection{Social Welfare}
We use the social welfare to define the network wide performance expressed as the sum of the utilities of UEs and SNs \cite{Eliz2011}.
\begin{equation}
\label{eq:socialwelfare}
        \Gamma(\eta) = \sum_{p \in \mathcal{P}} \sum_{m \in \mathcal{L}_p} U_{p,m}(R_{p,m},w_{p,m}, \eta),
\end{equation}
\noindent
where $\mathcal{M}$ and $\mathcal{P}$ are the set of UEs and set of SNs in the network, respectively. The objective is to maximize the total network wide social welfare given in (\ref{eq:socialwelfare}). Unfortunately, maximizing the network-wide social welfare in a centralized manner requires large information exchange between all SCBSs and UEs in the network, calling for a distributed solution with minimum coordination. To address this issue, we group mutually-interfering SCBSs into a number of clusters such that SCBSs within a cluster coordinate locally among each other. Specifically, we consider that SCBSs are grouped into a set of well-chosen clusters $\mathcal{C} = \{C_1,C_2, ..., {C_{|\mathcal{C}|}}\}$. Let $\eta_c$ represents the user association (matching), such that $\eta_c(m,p)$ represents the matching of UE $m$ and SN $p$ within cluster $c \in \mathcal{C}$. Each cluster $c$ consists of locally-coupled SCBSs in terms of mutual interference in which $\mathcal{N}_c$ denotes the number of SCBSs belonging to cluster $c\in \mathcal{C}$. It is assumed that, in cluster $c$, SCBSs efficiently offload traffic among each other while satisfying UEs' QoS. Moreover, the matching for each cluster is represented by a vector $\boldsymbol{\eta} = [\eta_1, \eta_2, ..., \eta_{|\mathcal{C}|}]$. Hereinafter, we refer to $\boldsymbol{\eta}$ as the ``network wide matching'', which captures the utilities of all the UEs and SNs in the network whereas the per cluster matching is denoted by $\eta_c$. Finally, we define the social welfare per cluster $c$, $\Gamma_c(\eta_c)$ given matching $\eta_c$ by:
\begin{equation}
\label{eq:socialwelfare-cluster}
     \Gamma_c(\eta_c) = \sum_{p \in \mathcal{P}_c} \sum_{m \in \mathcal{M}_c} U_{p,m}(R_{p,m},w_{p,m}, \eta_c),
\end{equation}
\noindent
where $\mathcal{M}_c$ is the set of the UEs, $\mathcal{N}_c$ the set of SCBSs, and $\mathcal{P}_c$ the set of SNs belonging to cluster $c \in \mathcal{C}$. The objective is to maximize the social welfare for all clusters, which is given by the following optimization problem:
\begin{subequations}
    \label{eq:mainprb}
    \begin{eqnarray}
        \underset{\boldsymbol{\eta}, \mathcal{C}}{\text{maximize}} && \sum_{\forall c \in \mathcal{C}} \Gamma_{c}(\boldsymbol{\eta}) \label{eq:matching1} \\
        \text{subject to} && |\mathcal{N}_c| \geq 1,\; \forall c \in \mathcal{C}, \label{eq:matching2} \\
        && \bigcup_{\forall c \in \mathcal{C}} \mathcal{N}_c = \mathcal{N}, \; \mathcal{N}_c \cap \mathcal{N}_{c'} = \emptyset, \; \nonumber \\ 
        &&\hspace{5em} \forall c,c' \in \mathcal{C},\; c \neq c',  \label{eq:matching3} \\
        && \sum_{\forall p \in \mathcal{P}_c} \eta_c(m,p) = 1, \; \forall m \in \mathcal{M}_c, \label{eq:matching4}
    \end{eqnarray}
\end{subequations}
\noindent
where constraints (\ref{eq:matching2}) and (\ref{eq:matching3}) imply that any SCBS is part of one cluster only. The constraint given in (\ref{eq:matching4}) depicts that a given UE $m$ can be matched to only one SN $p$ whereas, SN $p$ can be matched to one or more UEs for a given matching $\eta_c$. Solving (\ref{eq:mainprb}), requires global network information, which can be complex and not practical. Therefore, in the subsequent section, we propose a distributed solution composed of: 1) dynamic SCBS clustering, 2) flexible user association based on intra-cluster coordination. The different steps of our proposed solution are summarized in Fig. \ref{fig:flowchart}.
\begin{figure}[t]%
\centering
\includegraphics[width = \columnwidth]{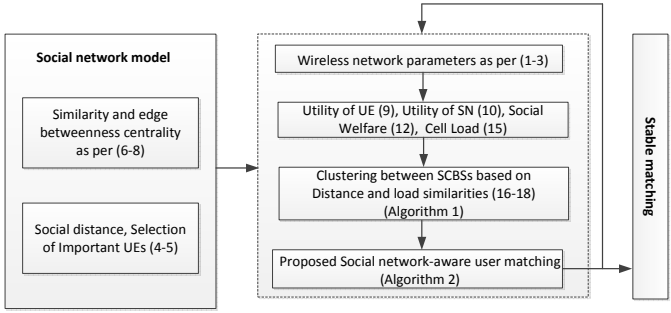}
\caption{An Illustration of the different steps of the proposed solution.}
\label{fig:flowchart}
\end{figure}%

\section{Dynamic Clustering}
\label{sec:4}
The centralized optimization problem in (\ref{eq:mainprb}) is difficult to solve and is combinatorial in nature. Developing a decentralized approach based on minimal coordination between neighboring SCBSs is needed. First, we propose a cluster-based mechanism which incorporates, both location of SCBS and their traffic load. Clustering enables coordination among well selected pairs of SCBSs. We propose a dynamic clustering approach, in which the cluster size varies dynamically depending on the dynamic nature of traffic (e.g., load, interference). Subsequently, we propose a distributed and self-organized matching algorithm to dynamically optimize the user association per cluster. The procedure only depends on the local information available at the cluster level. The set of SCBSs are partitioned into $|\mathcal{C}|$ non-overlapping clusters, such that:
\begin{figure*}[t]%
	\centering
	\includegraphics[scale = 0.60, angle=0]{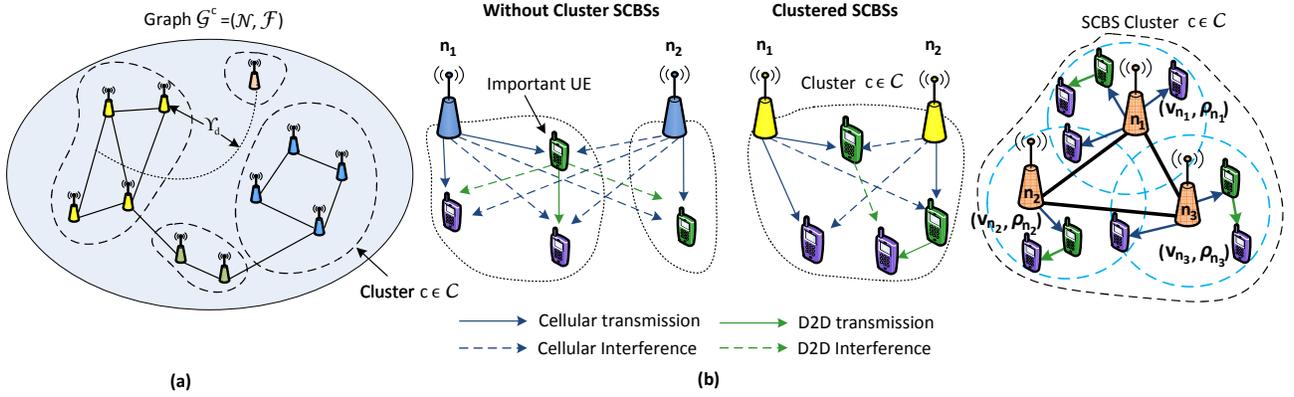}
	\caption{(a) Graph representation $\mathcal{G}^c=(\mathcal{N},\mathcal{F})$ of clustering among SCBSs, (b) Intra-cell and D2D interference for uncoordinated SCBSs, cluster or cluster-based SCBSs and D2D users interfere with each other within a cluster.}
	\label{fig:clustering}
\end{figure*}%
\begin{equation}
\label{eq:static-cluster}
     \bigcup_{\forall c \in \mathcal{C}} \mathcal{N}_c = \mathcal{N} \; \text{and} \; \mathcal{N}_c \cap \mathcal{N}_{c'} = \emptyset, \; \forall \; c \neq c'.
\end{equation}
Let $\mathcal{G}^c = (\mathcal{N},\mathcal{F})$ be the undirected connected graph, where $\mathcal{N}$ is the set of SCBSs $N$ and $\mathcal{F} \subset \mathcal{N} \times \mathcal{N}$ is the set of links between locally-coupled SCBSs. In order to calculate the cell load, let us  denote $\eta_n$ as a UE random association\footnote{Equivalently, the UE can be initially associated to the closest SCBS.} to an SCBS $n$ and $0\leq \rho_n(\eta_n) \leq1$ as the normalized load of SCBS $n\in\mathcal{N}$, given by:
\begin{align}
\label{eq:load-scbs}
	   \rho_n(\eta_n) \triangleq \sum_{\forall m \in \mathcal{L}_{n} \setminus \mathcal{M}_i } \frac{R_{n,m}}{R_{n,m}^{\text{max}}}, \forall i \in \mathcal{I},
\end{align}
where $R_{n,m}^{\text{max}}$ is calculated neglecting the interference from other SCBSs. The average load of each cluster $\rho_c$ for a given matching $\eta_c$ is the arithmetic average load of its member SCBSs, such that $\rho_c(\eta_c) = \frac{1}{|\mathcal{N}_c|}\sum_{\forall n \in \mathcal{N}_c} \rho_n(\eta_n)$. The clustering mechanism between SCBSs and intra-cluster coordination are demonstrated in Fig. \ref{fig:clustering}.
\vspace{-0.3cm}
\subsection{Similarity-based SCBS Clustering}
In order to minimize the signalling overhead, we group SCBSs based on similar attributes. There are numerous aspects that impact interference between SCBSs. Two key factors are their physical distance separation and traffic load condition. Having said that, we utilize location and traffic load similarities to group SCBSs and we use a spectral clustering algorithm \cite{ULuxburg2007} to identify similarities between SCBSs to form clusters. Next, we calculate the Gaussian affinity matrix \cite{Livehoods2012} representing the similarities between SCBSs based on their geographical locations and loads.

Let $v_{n_1}$ and $v_{n_2}$ be the geographical coordinates of SCBS $n_1$ and $n_2$, respectively, in the Euclidean space. Here, we define parameter $\Upsilon_d$ to represent the presence of a link or edge $f\in \mathcal{F}$ between neighboring SCBS $n_1$ and $n_2$ such that $\{f_{v_{n_1},v_{n_2}}=f_{v_{n_2},v_{n_1}}=1, ||v_{n_1} - v_{n_2}|| \leq \Upsilon_d\}$.
To find locally-coupled SCBSs in terms of distance, let $\boldsymbol{D}$ denotes the Gaussian distance similarity matrix, and let $ s_{n_1,n_2}^d$ be an element of $\boldsymbol{D}$ representing the distance similarity among SCBSs $n_1, n_2 \in \mathcal{N}$ given \cite{ULuxburg2007}:
\begin{equation}
\label{eq:distanceSim}
			s_{n_1,n_2}^d = \begin{cases}	\exp \bigg( \frac{-||v_{n_1} - v_{n_2}||^2}{2\sigma_d^2} \bigg), & \mbox{if $||v_{n_1} - v_{n_2}|| \leq \Upsilon_d$},\\
                            0, & \mbox{otherwise},
                            \end{cases}
\end{equation}
where the parameter $\sigma_d$ controls the impact of neighborhood size. For a given $\Upsilon_d$ the range of the Gaussian distance similarity for any two connected SCBSs is $[e^{-\Upsilon_d / 2\sigma_d^2},1]$, whereas the lower bound is determined by $\sigma_d$. The rationale for (\ref{eq:distanceSim}) is that, when the SCBSs are located far from each other, the distance similarity is low. On the other hand, the distance similarity increases as SCBSs come closer to one another and more likely to cooperate with each other.

Unlike the static distance based clustering in (\ref{eq:distanceSim}), the traffic load of SCBSs varies over time thus, the load based clustering provides a more dynamic manner of grouping neighboring SCBSs. Therefore, we are interested in clustering SCBSs which have load dissimilarities. let $ s_{n_1,n_2}^l$ be an entry of the Gaussian load dissimilarity matrix $\boldsymbol{L}$ between SCBS $n_1, n_2\in \mathcal{N}$ with respect to cell load $\rho_{n_1}$ and $\rho_{n_2}$, which is given \cite{ULuxburg2007}:
\begin{equation}
\label{eq:loadSim}
    s_{n_1,n_2}^l = \exp \bigg(\frac{||\rho_{n_1} - \rho_{n_2}||^2}{2\sigma_l^2} \bigg),
\end{equation}
where the parameter $\sigma_l$ controls the impact of load on the similarity. The range of load dissimilarity is $[1 , e^{1/2\sigma_l^2}]$. The upper bound of the dissimilarity is based on the choice of $\sigma_l$. We use a spectral clustering algorithm (Algorithm \ref{algo:algo1}) to form clusters between SCBSs based on their Gaussian affinity matrix. The Gaussian affinity matrix effectively captures the distance and load similarities. The Gaussian affinity matrix $\boldsymbol{Y}$ whose element $y_{n_1,n_2}$ represents joint similarity between two SCBSs $n_1, n_2 \in \mathcal{N}$ based on the distance and load is:
\begin{align}
	\label{eq:gaussiancluster}
	 \boldsymbol{Y}=\boldsymbol{D}^{\Omega} \cdot \boldsymbol{L}^{1-\Omega},
\end{align}
where $0 \leq \Omega \leq 1$ controls the impact of distance and load similarities on the joint similarity. Here, the cooperation between SCBSs is only possible if a physical link between them exists i.e, $\forall \Omega \in [0,1], f_{V_{n_1},v_{n_2}}=0 \Longrightarrow y_{n_1,n_2}=0$.
\begin{algorithm}[t]
	\footnotesize
	\caption{Spectral clustering for clustering SCBSs \cite{ULuxburg2007}}
	\label{algo:algo1}
	\DontPrintSemicolon 
	\KwIn{$\mathcal{N}, \boldsymbol{Y} = {y_{n_1,n_2}}, \mathcal{G}^c$ the graph of SCBS, $k_{\textrm{min}}$, $k_{\textrm{max}}$}
	Compute diagonal degree matrix $\boldsymbol{H}$ with diagonal $(d_1, ..., d_{nv})$ where $d_i = \sum_{j=1}^{nv} y_{n_i,n_j}$. \;
	$\boldsymbol{Z} := \boldsymbol{H} - \boldsymbol{Y}$ \;
	$\boldsymbol{Z}_{\textrm{norm}} := \boldsymbol{H} ^{-1/2} \boldsymbol{Z} \boldsymbol{H} ^ {-1/2}$. \;
	Let $\lambda_1 \leq .. \leq k_{\textrm{max}}$ be the smallest eigenvalues of $\boldsymbol{Z}_{\textrm{norm}}$. Set $k = \arg \max_{i=k_{\textrm{min}}, ..., k_{\textrm{max}}-1} \Delta_i$ where $\Delta_i = \lambda_{i+1}- \lambda_i $. \;
	find the $k$ smallest eigenvectors $e_1, ..., e_k$ of $\boldsymbol{Z}_{\textrm{norm}}$. \;
	Let $\boldsymbol{E}$ be an $nv \times k$ matrix with $e_i$ as columns. \;
	Use k-means clustering to cluster the rows of matrix $E$. \;
	Cluster set $\{1, ..., C_{|\mathcal{C}|}\}$. \;
\end{algorithm}
In our model, when the UEs are first admitted in the system, they will be associated to the SCBS based on the max-RSSI criterion. The SCBS to which UEs associate is then referred as \emph{anchor SCBS} i.e., UE $m$ associates with anchor SCBS $n$ if and only if $\text{RSSI}_{n,m} \geq \text{RSSI}_{n',m}$ for all $n'\in\mathcal{N}$. Moreover, we would like to stress the fact that the goal of clustering is to enable coordination among well selected pairs of SCBSs within the same cluster. Once the clusters are formed among SCBSs, UE $m \in \mathcal{M}_c$ will always be a part of the same cluster $c$, which corresponds to its anchor SCBS $n \in \mathcal{N}_c$. For UEs at the edge of multiple clusters, they will simply remain associated to the SCBS cluster that contains their original anchor SCBS to which they associated based on the RSSI criterion. Let $n(m)$ be the anchor SCBS of UE $m$ and $\mathcal{L}_n$ be the set of UEs served by SCBS $n$. Irrespective of the fact that whether UE $m$ is at the cluster edge or not, the following conditions will be always satisfied.
	\begin{enumerate}[(i)]
		\item $n(m) \in c \iff m \in \mathcal{M}_c$ such that $c = {\mathcal{N}_c \cup \mathcal{M}_c} $,
		\item $m \in \mathcal{L}_n, n \in \mathcal{N}_c \Rightarrow n(m) \in c \in \mathcal{C}$.
	\end{enumerate}
The above conditions imply that, a given UE $m$ will be associated to an SCBS $n$ based on the max-RSSI if and only if $m$ and $n$ belong to the cluster $c$. Furthermore, it implies that the set of UEs that are serviced by SCBS $n$ also belong to the same cluster $c$. It is worth to mention that after clustering is performed, UEs can be served by any SN (i.e., SCBS, important UE) belonging to the same cluster $c$ based on the proposed association within cluster $c$.

\section{Per-Cluster Social Network-Aware User Association as a Matching game with Externalities}
\label{sec:5}
Our objective is to develop a self-organizing mechanism for solving (\ref{eq:mainprb}). In order to overcome the combinatorial nature of the user association problem, we make use of the framework of matching theory in which, the social and wireless characteristics are incorporated into the matching game. Such wireless and social effects motivate the need for advanced model for matching theory that take into account the wireless interference and strength of social ties. Thus, we propose a social network-aware matching game per cluster $c \in \mathcal{C}$ capturing both physical and social aspects of the network in which each UE $m \in \mathcal{M}_c$ is associated to the best serving node $p \in \mathcal{P}_c$ via a matching $\eta_c: \mathcal{M}_c \rightarrow \mathcal{P}_c$.
\begin{defn}
\label{def:1}
A \textit{matching game} is defined by two sets of players ($\mathcal{M}_c, \mathcal{P}_c$) and two preference relations $\succ_m$, $\succ_p$ for each UE $m \in \mathcal{M}_c$ to build his preference over SN $p \in \mathcal{P}_c$ and vice-versa in a cluster $c$. The outcome of the matching game is the association mapping $\eta_c$ that matches each player $ m \in \mathcal{M}$ to player $p=\eta_c(m)\; p\in \mathcal{P}_c$ and vice versa such that $m = \eta_c(p)\;, m \in \mathcal{M}_c$.
\end{defn}
A \emph{preference relation} $\succ$ is defined as a reflexive, complete and transitive binary relation between players in $\mathcal{M}_c$ and $\mathcal{P}_c$. Thus, a preference relation $\succ_m$ is defined for every UE $m \in \mathcal{M}_c$ over the set of SNs $\mathcal{P}_c$ such that for any two nodes in $p,\tilde{p} \in \mathcal{P}^{2}_c, p \neq \tilde{p} $ and two matchings $\eta_{c}, \eta'_{c} \in \mathcal{M}_c \times \mathcal{P}_c, \eta_c \neq \eta_c' \;, p =\eta_c(m)\;, \tilde{p}=\eta'_{c}(m)$:
\begin{align}
    \label{eq:ue-pref}
   & (p,\eta_c, \boldsymbol{\eta}_{-c}) \succ_m (\tilde{p},\eta'_c, \boldsymbol{\eta}_{-c}) \Leftrightarrow \nonumber \\ 
   & U_{p,m}(R_{p,m}, w_{p,m},\eta_c, \boldsymbol{\eta}_{-c}) > U_{\tilde{p},m}(R_{\tilde{p},m},w_{\tilde{p},m}, \eta_c', \boldsymbol{\eta}_{-c}),
\end{align}
where  $(p,m) \in \eta_c$ and $(\tilde{p},m) \in \eta_c'$. Similarly the preference relation $\succ_p$ for SN $p$ over the set of UEs $\mathcal{M}_c$ is defined such that for any two UEs $m,\tilde{m} \in \mathcal{M}_c, m \neq \tilde{m}\;, m =\eta_c(p)\;, \tilde{m}=\eta'_{c}(p)$:
\begin{equation}
 \label{eq:sn-pref}
    (m,\eta_c, \boldsymbol{\eta}_{-c}) \succ_p (\tilde{m}, \eta_c', \boldsymbol{\eta}_{-c}) \Leftrightarrow U_{p}(\eta_c) > U_{\tilde{p}}(\eta_c').
\end{equation}
Hereinafter, for notational simplicity we define $U_{p,m}(\eta_c) :\triangleq U_{p,m}(R_{p,m}, w_{p,m},\eta_c, \boldsymbol{\eta}_{-c})$.
\begin{remark}
    The proposed social network-aware matching game has externalities and peer effects.
\end{remark}
\noindent
Each SN and UE independently rank one another based on the respective utilities in (\ref{eq:utility-ue}) and (\ref{eq:utility-scbs}) that capture the interference and social ties among nearby UEs. However, the selection preferences of UEs are \emph{interdependent} and influenced by the existing network wide matching, which leads to a many-to-one matching game. Such effects which dynamically change the preference of each player in the network, are called externalities \cite{Bando2012}. In particular, the considered game is a matching game with \emph{externalities} due to mutual interference and social ties between nodes, which differs from classical applications of matching theory in wireless such as those in \cite{EA2011, SBayat2012, ALeshem2011}. Thus, each player $m \in \mathcal{M}_c$ has a preference over players in $p \in \mathcal{P}_c$ and vice versa and these preferences change as the game evolves. Finally, each UE is matched (associated) to one SN, while SNs can be matched to multiple UEs which makes the matching game as many-to-one.

In classical matching games with no peer effects, each UE has a strict preference over SNs and vice versa that remain unchanged for the overall game. The key premise of our work is that peer effects are often the result of an underlaying social network. For our work we assume that peer effects is captured at the important UEs due to social ties with other UEs in the proximity as per (\ref{eq:utility-ue}). The strength of social ties (i.e, social distance) among players may change if UEs are socially connected to other UEs within their proximity range and thus, impact on the preference at the UEs and important node (SNs). To deal with externalities and peer effects due to the interference and social ties, the most important notion is the stability of the solution. To solve the problem in (\ref{eq:mainprb}), each UE and SN defines its preference over each other using (\ref{eq:ue-pref}) and (\ref{eq:sn-pref}). The objective of each player is to maximize its own utility, by associating to its most preferred SN.
\subsection{Proposed Social Network-Aware User Association Algorithm}
\label{sec:51}
In order to solve the proposed matching game, usually deferred acceptance algorithm guarantees a stable solution in one-to-one matching \cite{ALeshem2011, EA2011, YWu2011}. Nevertheless, such approaches do not account for externalities and peer effects, and, thus, they may yield lower utilities or may not converge. In fact, due to externalities and peer effects, players continuously change their preference orders, in response to the formation of other UE-SN links which renders classical deferred acceptance solutions such as in \cite{ALeshem2011, EA2011, YWu2011} not applicable for our model. Therefore, to seek a stable user association an Algorithm \ref{algo:algo2} is proposed which is based on the concept of Markov Chain Monte Carlo (MCMC) \cite{Eliz2011}. In this approach, instead of using the greedy way of selecting the ``best'' matching, the matching is chosen based on a probability, which depends on the swap resulting in an increase of the social welfare for a given cluster $c$.

\LinesNumberedHidden{
\begin{algorithm}[t]
\footnotesize
\caption{Proposed Social Network-aware User Association Algorithm}
\label{algo:algo2}
\DontPrintSemicolon 
\KwData{Each UE $m$ is initially associated to a randomly selected SCBS $n$.}
\KwResult{Convergence to a stable matching $\boldsymbol{\eta}$.}

\textbf{Phase I - Social distance computation;}

\begin{itemize}
    \item UEs and SNs exchange social-aware information and compute $\boldsymbol{S}$ and $\boldsymbol{A}$ using \eqref{eq:similarity} and \eqref{eq:betweenness}; 
	\item Calculation of important UEs list $\mathcal{I}$ based on the social distance $\boldsymbol{W}$ using (\ref{eq:socialdistance}) and (\ref{eq:weightcost});
    \item Node with highest rank in the sorted list $\mathcal{I}_n$ is selected as the important node $i \in \mathcal{I}$
\end{itemize}

\While{$t \leq  t_{\text{max}}$} {

\textbf{Phase II - Clustering among SCBSs;}
\begin{itemize}
    \item Compute gaussian distance and load similarity metrics in (\ref{eq:distanceSim}), (\ref{eq:loadSim});
    \item Gaussian similarity matrix computed using (\ref{eq:gaussiancluster});
    \item Clusters $|\mathcal{C}|$ are formed among SCBSs using Algorithm \ref{algo:algo1};
\end{itemize}

\textbf{Phase III - SN discovery and utility computation;}

\begin{itemize}
	\item Each UE $m$ discovers a SN $p$ in the cluster vicinity $c \in \mathcal{C}$;
	\item $U_{p,m}(R_{p,m}, w_{p,m})$ using (\ref{eq:utility-ue}), $U_{p}$ using (\ref{eq:utility-scbs}), and social welfare $\Gamma_{c}(\eta_c)$ using (\ref{eq:socialwelfare-cluster}) for cluster $c$ are updated;

\end{itemize}

\textbf{Phase IV - Swap-matching evaluation;}
   \While{$c \leq \max(|\mathcal{C}|)$} {
   \begin{itemize}
	   \item Pick a random pair of UEs $\{m,\tilde{m}\} \in \mathcal{M}_c$ within the cluster $c$;
    \end{itemize}
    	\While{$count \leq  count_{\textrm{max}}$} {
    	\begin{itemize}
    			 \item $U_{p,m}(R_{p,m}, w_{p,m})$, $U_{p}$ are updated based on the current \\ matching $\eta_c$;
    				\item UEs and SNs are sorted by $\Gamma_{c}(\eta_c)$;
        			\item swap the pair of UEs $\eta_c \Rightarrow \eta_c^{\leftrightarrow}$
    				\item $\Gamma_{c}(\eta_c, \boldsymbol{\eta}_{-c}) = \Gamma_{c}^{\textrm{best}}(\eta_c, \boldsymbol{\eta}_{-c})$
    	\end{itemize}
         $P_T = \frac{1}{1+e^{-\vartheta (\Gamma_{c}(\eta_c, \boldsymbol{\eta}_{-c}) - \Gamma_{c}(\eta_c^{\leftrightarrow}, \boldsymbol{\eta}_{-c}))}}$;\\
         $(\eta_c, \boldsymbol{\eta}_{-c}) \leftarrow (\eta'_c, \boldsymbol{\eta}_{-c})$ change the configuration with probability $P_T$;\\
    		    \If{$\Gamma_{c}(\eta_c^{\leftrightarrow}, \boldsymbol{\eta}_{-c}) > \Gamma_{c}^{\textrm{best}}(\eta_c,\boldsymbol{\eta}_{-c})$}{
                        $\Gamma_{c}^{\textrm{best}}(\eta_c,\boldsymbol{\eta}_{-c}) = \Gamma_{c}(\eta_c^{\leftrightarrow}, \boldsymbol{\eta}_{-c})$ \;
        }
    \Else{
      SN $p$ refuses the proposal, and UE $m$ sends a proposal to the next configuration at $count$ \;
    }
		$count = count + 1$
    }
    $c = c + 1$
}
$t = t + 1$
}
\textbf{Phase V - Stable matching}
\end{algorithm}}

In the proposed algorithm, an important UE $i$ and set of UEs serviced by the important UE $\mathcal{M}_i$ seek the same content. The preferences of both the UEs and SNs are done locally within a given cluster $c$, whereas the coordination is required between adjacent SCBSs. If UE $m \in \mathcal{M}_c$ is not currently served by its most preferred SN $p \in \mathcal{P}_c$, it sends a matching proposal to another SN $\tilde{p}$. Upon receiving a proposal, SN $\tilde{p}$ updates its utility and accepts the request of the UE if the externalities and peer effects resulting from such swap do not yield a degradation of the social welfare of the cluster. The main goal of each UE is to maximize its own utility while associating with the most preferred SN or important UE. Initially, each UE is associated to a randomly selected SN based on the max-RSSI criterion. In the first phase, social distance matrix is calculated using (\ref{eq:socialdistance}) and then each SCBS compute the list of important UEs $\mathcal{I}_n$ using (\ref{eq:weightcost}). In the next phase, clusters are formed among SCBSs based on their gaussian distance and load similarity using Algorithm \ref{algo:algo1}. Then, the utilities of all players and social welfare of a given cluster is calculated for the current matching $\eta_c$. In the fourth phase, UEs and SNs update their respective utilities and individual preferences over one another. Subsequently, at each iteration, a chosen UE pair is swapped with a probability that depends on the change in the cluster's social welfare: a positive change in the social welfare of a cluster yields a probability of swapping larger then $1/2$ and vice-versa. As a result, the algorithm does not get caught in a local optimum and the algorithm continuously keeps track of the ``best'' matchings. Algorithm \ref{algo:algo2} terminates when no further improvement can be achieved. After phase IV, the algorithm converges to a local maximum of the social welfare for a given cluster. The Algorithm \ref{algo:algo1} continue until it reaches to stable matching.
\subsection{Convergence and Stability}
\label{sec:52}
The concept of peer effect and externalities requires us to adopt a new stability concept based on the idea of ``pairwise stability'' \cite{Eliz2011}. Before defining the pairwise stability, we first define a swap matching.
\begin{defn}
\label{def:2}
A \emph{swap matching} is formally defined as $\eta_c^{m \leftrightarrow \tilde{m}} = \{ \eta_c \setminus \{(p,m),(\tilde{p},\tilde{m}) \}\} \cup \{ (m,\tilde{p}),(\tilde{m},p)\}$. In each swap two UEs change their matching with their respective SNs while other matchings remain fixed. Having defined swap matching, we further define pairwise stability.
\end{defn}
\begin{defn}
\label{def:3}
Given a matching $\eta_c$, a pair of UEs $m,\tilde{m}$ and SNs $p, \tilde{p}$ within a cluster $c$, a pairwise matching is \emph{stable} if and only if there does not exist a pair of UEs $(m,\tilde{m})$ such that:
\begin{enumerate}[(i)]
  \item $\forall y \in \{m,\tilde{m}, p,\tilde{p} \}$, such that $U_{y,\eta_c^{m \leftrightarrow \tilde{m}}(y)}(\eta_c) \geq U_{y, \eta_c(y)}(\eta_c)$ and
  \item $\exists y \in \{m,\tilde{m},p,\tilde{p}\}\; U_{y,\eta_c^{m \leftrightarrow \tilde{m}}(y)}(\eta_c) > U_{y, \eta_c(y)}(\eta_c)$.
\end{enumerate}
\end{defn}
\noindent
A matching $\eta_c$ is said to be pairwise stable if there does not exist any UE $\tilde{m}$ or SN $\tilde{p}$, for which SN $p$ prefers UE $\tilde{m}$ over UE $m$ or any UE $m$ which prefers SN $\tilde{p}$ over $p$. From Definition \ref{def:3}, we can see that if two UEs swap between two SNs, the SNs involved in the swap must ``approve'' the swap. Similarly, if two SNs want to swap between two UEs, the UEs and SNs must agree to the swap. This definition is useful for proving the two-sided stability of our proposed matching.  Next, we will show that two-sided pairwise stable matching will always exist in our game. For that, we assume that neither UEs nor SNs can remain unmatched in the cluster, and we restrict ourselves to considering swap of UEs between SNs. However, before stating this result, we require the following Lemma:
\begin{lem}
\label{lem:1}
\emph{Any swap matching $(\eta_c^{m \leftrightarrow \tilde{m}})$ for which}
\begin{enumerate}[(i)]
  \item $\forall y \in \{m,\tilde{m}, p,\tilde{p} \}$, such that $U_{y,\eta_c^{m \leftrightarrow \tilde{m}}(y)}(\eta_c) \geq U_{y, \eta_c(y)}(\eta_c)$ and
  \item $\exists y \in \{m,\tilde{m},p,\tilde{p}\}\; U_{y,\eta_c^{m \leftrightarrow \tilde{m}}(y)}(\eta_c) > U_{y, \eta_c(y)}(\eta_c)$,
\end{enumerate}
yields $\Gamma_c(\eta_c^{m \leftrightarrow \tilde{m}}) > \Gamma_c(\eta_c)$ which can be written as $\Gamma_c(\eta_c^{\leftrightarrow}) > \Gamma_c(\eta)$.
\end{lem}
\begin{proof}
The symmetry of social network and swap matching are key factors for guaranteeing pairwise stability. Moreover, the approved swap in the symmetry social network results in a Pareto improvements for the players involved in the swap, as clearly seen from the definition of pairwise stability. For all other players, a non-negative change in utility follows from the symmetry of the social-graph. We assume a UE centric matching in which UEs have preference over the SNs and not vice versa. Thus, for proving Lemma \ref{lem:1}, we consider only one-sided matching game rather than two-sided matching. In order to proof Lemma \ref{lem:1}, we start by calculating the difference in the social welfare for the swap matching $\eta_c^{\leftrightarrow}$ and given matching $\eta_c$. Without loss of generality, it is assumed that the swapping of UE $m$ strictly increases its utility. Define $\eta_c(m)=p$, and $\eta_c(\tilde{m})=\tilde{p}$, then let us start by calculating the change in the utility of UE $m$ which is given by:
\begin{equation}
\label{eq:proof1}
     0 < U_{p,m}(\eta_c) - U(\eta_c^{\leftrightarrow}) = \sum_{z \in \eta_c(\tilde{p})} U_{z,m} - U_{\tilde{m},m}  - \sum_{z \in \eta_c(p)} U_{z,m}.
\end{equation}
Similarly, for the change in utility for UE $\tilde{m}$, which is given by:
\begin{equation}
\label{eq:proof2}
    0 \leq U_{p,m}(\eta_c) - U(\eta_c^{\leftrightarrow}) = \sum_{z \in \eta_c(p)} U_{z,\tilde{m}} - U_{\tilde{m},m} - \sum_{z \in \eta_c(\tilde{p})} U_{z,\tilde{m}}.
\end{equation}
Adding the above inequalities (\ref{eq:proof1}), (\ref{eq:proof2}) we have:
\begin{multline}
\label{eq:proof3}
    0 < \sum_{z \in \eta_c(\tilde{p})} \bigg(U_{z,m} - U_{z,\tilde{m}}\bigg) \\
    + \sum_{z \in \eta_c(p)} \bigg( U_{z,m} - U_{z,\tilde{m}} \bigg) - 2U_{\tilde{m},m} := \delta_c.
\end{multline}
Consider a matching $\eta_c$ and a swap matching $\eta_c^{\leftrightarrow}$ that satisfies (i) and (ii) of Lemma \ref{lem:1}. The total change in the utility for all UEs $\mathcal{M}_c$ is:
\begin{align}
\label{eq:proof4}
     \bigtriangleup_{\mathcal{M}_c} &:= \sum_{z \in \mathcal{M}_c} U_{p,z}(\eta_c^{\leftrightarrow}) - \sum_{z \in \mathcal{M}_c} U_{p,z}(\eta_c) \nonumber \\
								    &:=\delta_c + \underbrace{ \sum_{z \in \eta_c(\tilde{p})} U_{z,m} - U_{z,\tilde{m}} }_\text{utility gain from $m$ associating $\tilde{p}$ } - \underbrace{ \sum_{z \in \eta_c(p)} U_{z,m} }_\text{utility loss from $m$ leaving $p$} \nonumber \\
								    & \quad +  \underbrace{ \sum_{z \in \eta_c(p)} U_{z,\tilde{m}} - U_{m,\tilde{m}} }_\text{utility gain from $\tilde{m}$ associating $p$}  - \underbrace{ \sum_{z \in \eta_c(\tilde{p})} U_{z,\tilde{m}} .}_\text{utility loss from $\tilde{m}$ leaving $\tilde{p}$}   \nonumber
\end{align}
\begin{align}
     & \bigtriangleup_{\mathcal{M}_c} := 2\delta_c > 0,
\end{align}
where (\ref{eq:proof4}) assumes that the social graph is symmetric. The total change in utility for all SNs $\mathcal{P}_c$ we have:
\begin{align}
\label{eq:proof5}
 & 0 \leq U_p(\eta_c^{\leftrightarrow}) - U_p (\eta_c) + U_{\tilde{p}}(\eta_c^{\leftrightarrow}) - U_{\tilde{p}}(\eta_c) := \bigtriangleup_{P_c}.
\end{align}
Without loss of generality, assume that UE $m$ strictly improves the utility while another UE $p$ either improves or is indifferent to the swap. It can be shown from (\ref{eq:proof5}) that the SNs $p$ and $\tilde{p}$ are affected by the swap with non-negative change in their utilities. Thus the social welfare strictly increases:
\begin{equation}
\label{eq:proof6}
			 \Gamma_c(\eta_c^{\leftrightarrow}) - \Gamma_c(\eta_c) = \bigtriangleup_{M_c} + \bigtriangleup_{P_c} > 0.
\end{equation}
\end{proof}
\noindent
Expanding on the idea presented in Lemma \ref{lem:1}, it is easy to prove the following theorem.
\begin{thm}
\label{thm:1}
 All local maxima of the social welfare for a cluster $c$ given in (\ref{eq:socialwelfare-cluster}) are two-sided pairwise stable.
\end{thm}

\begin{proof}
Let $\Gamma_c(\eta_c)$ be a local maximum of a given matching $\eta_c$. Lemma \ref{lem:1}, shows that any swap matching which is acceptable by all players satisfies conditions (i) and (ii), and strictly increases the social welfare of cluster $c$. Nevertheless this assumption is contradictory as $\eta_c$ is a local maximum for cluster $c$. Therefore, $\eta_c$ must be stable. It is worth nothing that, not all pairwise stable matchings are local maxima\footnote{This case can be considered when one player rejects a swap as its utility would decrease, but the other player gets benefit from such a swap. If forced swap happens, the total social welfare could increase, but only at the expense of the first player.} of $\Gamma_c(\eta_c)$.
\end{proof}
\begin{cor}
\label{cor:1}
The proposed Algorithm \ref{algo:algo2} is guaranteed to converge to a two-sided stable matching.
\end{cor}
\begin{proof}
It can be shown from Lemma \ref{lem:1} and Theorem \ref{thm:1} that the algorithm converges to a stable matching, since with each iteration the social welfare strictly improves, and all local maxima of $\Gamma_c$ are stable matchings. All swaps among players must be agreed upon as given in Lemma \ref{lem:1}. Moreover, UEs have limited transmission range and can be matched with a finite number of SNs in their vicinity, therefore the possible swaps for the players are \emph{finite}. Every UE has a finite number of choices to swap so we have a finite set of matching for a given number of the SNs. In addition, considering all the possible swaps each UE is associated to its most preferred SN and vice versa. Algorithm \ref{algo:algo2} terminates, when no further improvement in social welfare is achieved by all possible swaps among players. Therefore, Algorithm \ref{algo:algo2} converges to a two-sided stable matching after a finite number of iterations.
\end{proof}
\subsection{Complexity Analysis}
In order to compute the complexity of the proposed algorithm, we start with the simple case in which the matching game has no social ties (peer effects) and UEs have strict preference ordering. As matching is done per cluster, we compute the complexity and message overhead for the matching $\eta_c$ within one cluster $c \in \mathcal{C}$. We assume that $\Phi_p$ is the maximum number of UEs matched to SN $p \in \mathcal{P}_c$ and $\Phi_c$ is the total number of UEs matched per cluster $\Phi_c = \sum_{\forall p \in \mathcal{P}_c} \Phi_p$ such that $\Phi_p = \Phi_{\tilde{p}}, p \neq \tilde{p}$. The value of $\Phi_p$ depends on the available bandwidth. Let $\Phi_s(\leq \Phi_c)$ denotes the number of satisfied matched UEs, which is based on the preference ordering. For simplicity, lets assume $\Phi_s$ is constant during all the iterations such that $\Phi_{s} = \Phi_{\tilde{s}}, \forall s \neq \tilde{s}$. To compute the complexity per cluster, we consider two worst case scenarios, when all UEs $m \in \mathcal{M}_c$ inside a cluster $c$ are matched to a single SN: (1) when the number of UEs inside the cluster are less then the total number of matched UEs i.e., $M_c \leq \Phi_c, |\mathcal{M}_c|= M_c$ and, (2) when the number of UEs is greater than the total number of matched UEs i.e., $M_c > \Phi_c$. Our goal is to analyze the maximum number of iterations required for convergence and the maximum number of proposals sent from UEs to SN (message overhead) for both cases. In each iteration $t$, UEs send proposal to their most preferred SN (i.e., important node or SCBS), and the SN accepts or rejects the received proposal based on its preference ordering and available capacity. Therefore, the number of matched but unsatisfied UEs at each iteration is less or equal than the available capacity.

For the first case, when the algorithm converges, all UEs are matched to a single SN, since, SNs prefer any UE to being unmatched. It can be observed that the worst case scenario happens, if all UEs have the same preference ordering. Thus, at the end of each iteration $t$ we have $M_c - \Phi_s t$ unsatisfied UEs. All UEs are matched and satisfied when the maximum number of iterations $t_{\text{max}}$ is obtained. i.e., $M_c- \Phi_s t_{\text{max}}=0$. Hence, the complexity is of order $\mathcal{O}(M_c)$. Moreover, we have $M_c - \Phi_s t$ proposal messages at each iteration $T$, and the total overhead for sending such messages is given by:
\begin{equation}
\label{eq:complexity-case-1}
    \xi_{\max}:= \sum_{t=1}^{t_{\text{max}}}(M_c - \Phi_s t + \Phi_s) = \frac{M_c(M_c + \Phi_s)}{2 \Phi_s}.
\end{equation}
For the second case when $M_c > \Phi_c$, once the algorithm converges, there are $M_c-\Phi_c$ unallocated UEs. The worst case happens, if all UEs have same preference ordering. Hence at $t_{\text{max}}$ iteration we have $M_c - \Phi_s t_{\text{max}}$ UEs unallocated. The complexity of the order $\mathcal{O}(\Phi_c)$, and the messaging overhead is equal to:
\begin{equation}
\label{eq:complexity-case-2}
    \xi_{\max}: = \sum_{t=1}^{\Phi_c} (M_c - \Phi_s t).
\end{equation}

The complexity of the proposed social network-aware algorithm will further depend on the social distance matrix $\boldsymbol{W}$ and the spectral clustering. The social distance matrix $\boldsymbol{W}$ is computed only once whereas, the spectral clustering algorithm runs for finite number of iterations.
\section{Simulation Results}
\label{sec:6}
We consider a single macro-cell in which UEs and SCBSs are uniformly distributed over the area of interest. Transmissions are affected by distance dependent path loss according to 3GPP specifications \cite{3GPP}. We assume that there is no power control, and thus the power is uniformly divided between UEs. It is also assumed that the bandwidth $B$ is divided equally between the served UEs. For simulation, we assume that one UE is selected as important UE per SCN. The simulation parameters are given in Table \ref{tab:simulation}. The position of UEs is assumed to be static, distance dependent path loss model for D2D communication of LOS and NLOS $103.8 + 20.9\log_{10}(d[km])$, $145.4+37.5 \log_{10}(d[km])$ respectively, is considered. Furthermore, for the social network, the selection of important UE is based on static social information which is collected during the network setup phase. We use a common full-buffer traffic model for all UEs in our simulations. The performance of the social-aware approach is compared with the baseline classical association approaches (i.e., max-RSSI (single time slot) and random association). In the random association, important UEs are chosen randomly and UEs are randomly associated to SN within their D2D coverage radius. For the proposed social-aware UE-association with clustering approach a dynamic clustering method is used, in which the number of clusters dynamically changes. Moreover, all statistical results are averaged over a large number of independent runs and high dense network deployment.
\begin{table}[t]
	\caption{Simulation Parameters}
	\centering
	\begin{tabular}{| p{5.2cm} | p{2.4cm} |}
		\hline
		\textbf{Parameter} & \textbf{Value} \\ \hline
		Bandwidth (MHz)   & 5  \\ \hline          
		Area ($m^2$)    & 500 \\ \hline
		Noise power spectral density $N_0$ \cite{3GPP} (dBm/Hz) & -174 \\ \hline
		SCBS, D2D transmission radius (m)    & 50, 20  \\ \hline
		SCBS, UE transmission powers (dBm)   & 23, 15 \\ \hline 
		Tunable parameters,  $\alpha $, $\beta$    & 0.5 \\ \hline
		Inter-site distance (m) & $40$ \\ \hline
		$\tau_0$ such that $\tau_0 + \tau_1 = 1$ & $0.84$ \\ \hline
		Impact of load similarity $\sigma_l$ & 1 \\ \hline
		Impact of neighborhood size (distance) $\sigma_d$  & $100$ \\ \hline
		Parameter controls the impact of distance and load on similarity $\Omega$ & 0.5 \\ \hline
		$k_{\textrm{min}}$, $k_{\textrm{max}}$ parameters for clustering   & $2$,  $ \lceil (N/2)+1 \rceil$ \\ \hline
		cluster radius (m)  & $200$ \\ \hline
		Boltzman temperature $\vartheta$  & $\vartheta = 1 - \big( \frac{count}{count_{\textrm{max}}} \big)$ \\
		\hline
	\end{tabular}
	\label{tab:simulation}
\end{table}
\begin{figure}[t]%
	\centering%
	\includegraphics[width = \columnwidth]{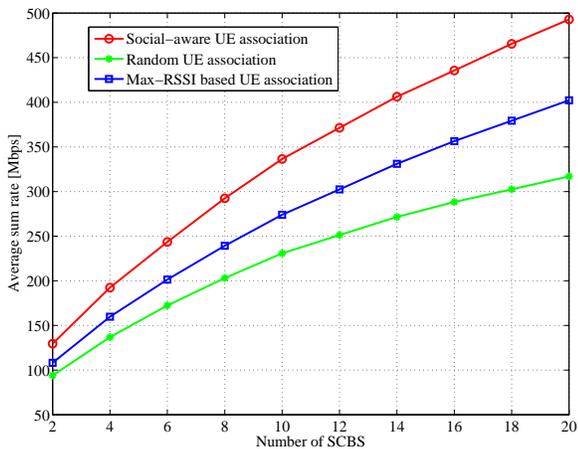}
	\vspace{-0.45cm}
	\caption{Average sum rate for a fixed number of UEs ($M=10$) per SCBS, for the proposed and baseline approaches.}
	\label{fig:avg-sum-rate-vs-scbs}
\end{figure}%
\vspace{-0.3cm}
\subsection{Impact of SCBS Density}
Fig. \ref{fig:avg-sum-rate-vs-scbs} shows the average sum rate as a function of the density of SCBSs $N$, and fixed number of UEs per SCBS $M=10$. Fig. \ref{fig:avg-sum-rate-vs-scbs} clearly shows that, in the proposed social-aware approach, user association improves the sum rate. In particular, we can see that, as the number of SCBSs increases, the average sum rate increases. This is due to the fact that, an increase in the number of SCBSs increases the number of important UEs, hence UEs associate with an important UE or SCBS based on their respective utility. Fig. \ref{fig:avg-sum-rate-vs-scbs}, also shows that the performance gains in terms of sum rate for the proposed social-aware association approach increases when the number of SCBSs increases in the system. We further note that the proposed social-aware approach yields significant performance gains for all network size, reaching up to $23\%$ over the max-RSSI based approach and $56\%$ over the random UE association approach.
\begin{figure}[t]%
	\centering%
	\includegraphics[width = \columnwidth]{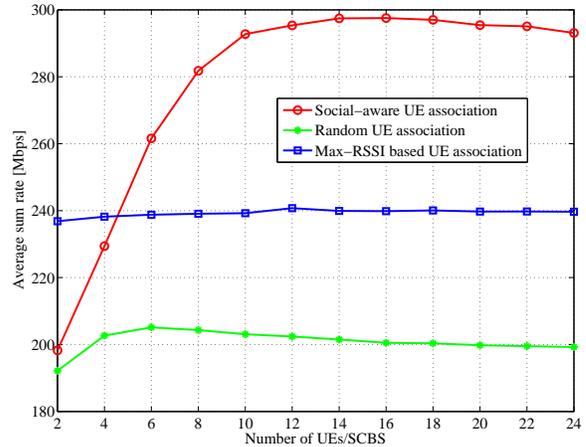}
	\vspace{-0.45cm}
	\caption{Average sum rate for a fixed number of SCBS ($N = 8$), under the considered approaches.}
	\label{fig:avg-sum-rate-vs-ue}
\end{figure}%
\begin{figure}[t]%
	\centering%
	\includegraphics[width = \columnwidth]{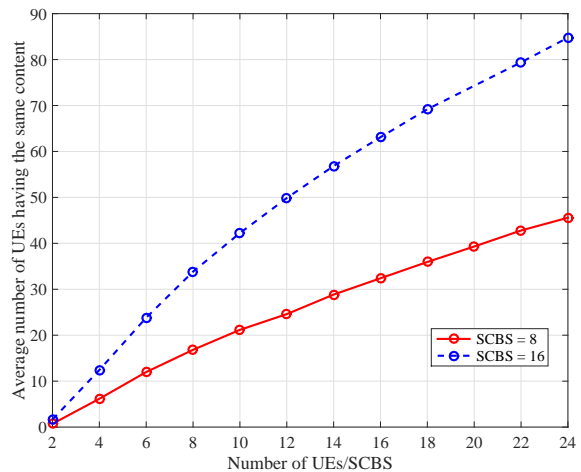}
	\vspace{-0.45cm}
	\caption{Average number of UEs having the same content versus density of UEs.}
	\label{fig:same-content}
\end{figure}%
\subsection{Impact of UE Density Per SCBS}
\begin{figure*}[t]
	\centering
	\begin{minipage}[t] {0.495\linewidth}
		\includegraphics[width=\columnwidth]{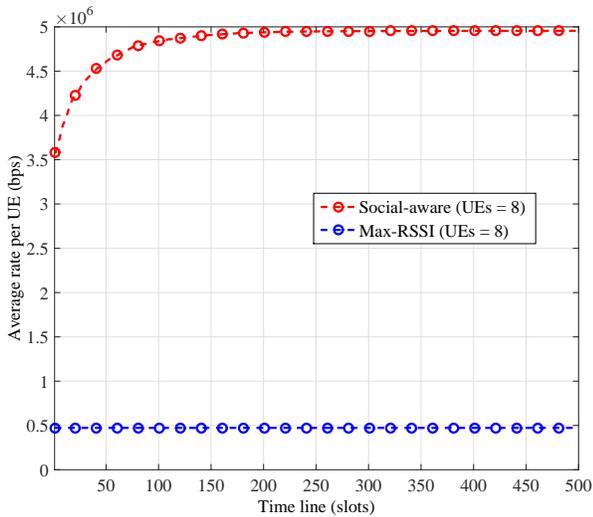}
	\end{minipage}
	\begin{minipage}[t]{0.495\linewidth}
		\includegraphics[width=\columnwidth]{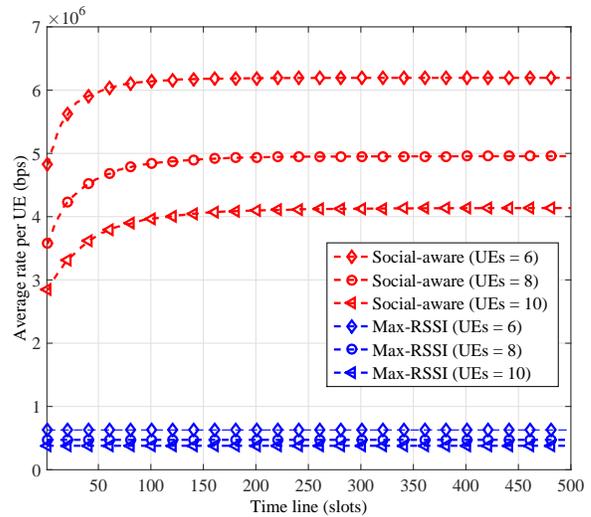}
	\end{minipage}
	\vspace{-0.8cm}
	\caption{Average achievable rate per UE for fixed number of SCBSs ($N=8$), under the considered approaches.}
	\label{fig:time-line-rate}
	\vspace{-0.6cm}
\end{figure*}
Fig. \ref{fig:avg-sum-rate-vs-ue} shows the average sum rate for a fixed number of SCBSs $N = 8$ and varying density of UEs. It is worth to mention that our proposed approach is suited for dense networks where large number of UEs per SCBS are deployed. It can be seen from the figure that, there is notable performance gain in terms of average sum rate for the proposed approach as compared to random and max-RSSI approach for varying number of UEs from $6$ to $24$ per cell. Moreover, it is also noted that in random UE association, UEs are associated to any SN within vicinity without consideration of RSSI and social-ties between UEs. From Fig. \ref{fig:avg-sum-rate-vs-ue}, we can also observe that, by increasing the number of UEs, the average sum rate increases up to $48\%$ and $24\%$ compared to random UE association and max-RSSI, respectively. Furthermore, we can see that with the fewer number of UEs per SCBS such as $2$ and $4$ UEs/SCBS the max-RSSI association approach outperforms over the proposed approach. This is due to the under utilization of the second time-slot $\tau_1$ if no D2D links are formed in neighboring SCBSs.\footnote{Note that the results are averaged over multiple realizations.}

Fig. \ref{fig:same-content} shows the average number of UEs having the same content for a fixed number of SCBSs $N = 8$ and $N=16$ with the variant density of UEs. It can be seen from the figure that, the number of UEs having the same content (i.e., size of social network) increases as the density of UEs increases.

Fig. \ref{fig:time-line-rate} shows the change in the average achievable rate per UE under the considered approaches. In order to examine the data rate per UE, we fixed the number of SCBSs $N=8$ and varied the number of UEs per SCBS. The average rate per UE is constant over the number of time slots in case of max-RSSI approach (single time slot). The achievable UE rate varies as a result of its association (matching) to SN (SCBS, or important UE). It can be observed from Fig. \ref{fig:time-line-rate}, that as we increases the number of UEs per SCBS, more time slots are required to achieve higher rate per UE. This is due to the fact that, with the increase in the number of UEs per SCBS, more D2D links can be exploited. Therefore, as the number of UEs increase in the system more time slots are required to find the suitable UE-SN association.
\vspace{-0.3cm}
\subsection{Cell Edge Performance for Fixed Number of UEs Per SCBS}
\begin{figure}[t]%
	\centering%
	\includegraphics[width = \columnwidth]{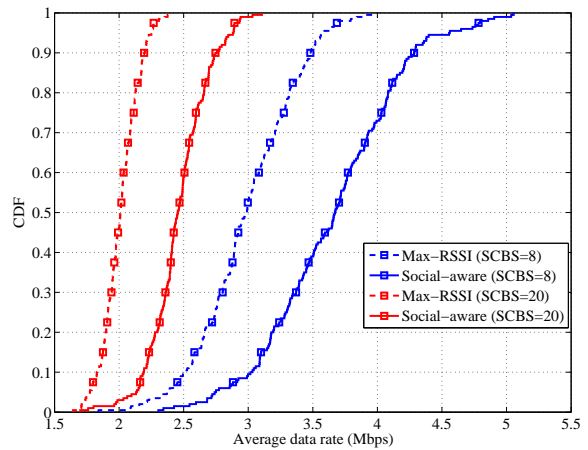}
	\vspace{-0.8cm}
	\caption{Cumulative density function of UE's rate for fixed UEs per SCBS ($M=10$), SCBSs ($N = 8$) and ($N = 20$) under the considered approaches.}
	\label{fig:cdf-scbs}
\end{figure}%

Fig. \ref{fig:cdf-scbs} shows the cumulative density function of UE's data rate for $M=10$ UEs per SCBS and different number of SCBS i.e., $N=8$ and $N=20$. In order to examine the gains in the UE's rate we analyze different percentile of user throughput. It can be shown from the figure that there is an increase in the UE data rate relative to the social unaware association approach. Table \ref{tab:cdf-scbs} shows the different percentiles of UE throughput for different for different density of SCBSs. We can see that, for $N=8$ SCBSs, our proposed social-aware user association outperforms the social-unaware user association (max-RSSI) approach by up to $21\%$, $22\%$ and $28\%$ for $5$-th, $50$-th and $95$-th percentiles of user throughput, respectively. Fig. \ref{fig:cdf-scbs} shows that, in case of $N=20$ the social-aware user association approach shows significant gains compared to social-unaware approach (max-RSSI) in terms of average data rate up to $16\%$, $23\%$ and $27\%$ for $5$-th, $50$-th and $95$-th percentiles of user throughput, respectively.
\begin{table}[!t]
	\caption{Percentiles of UE throughputs}
	\centering
	\begin{tabular}{| p{5.0cm} | p{0.7cm} | p{0.7cm} | p{0.7cm} |}
		\hline
		\textbf{Percentiles of UE throughputs} & \textbf{5-th} & \textbf{50-th} & \textbf{95-th} \\ \hline
		Average data rate gain as compare to max-RSSI approach ($N=8$)  & $21\%$  & $22\%$  & $28\%$  \\ \hline
		Average data rate gain as compare to max-RSSI approach ($N=20$) & $16\%$  & $23\%$  & $27\%$  \\ \hline
	\end{tabular}
	\label{tab:cdf-scbs}
\end{table}
\begin{figure}[t]%
	\centering%
	\includegraphics[width = \columnwidth]{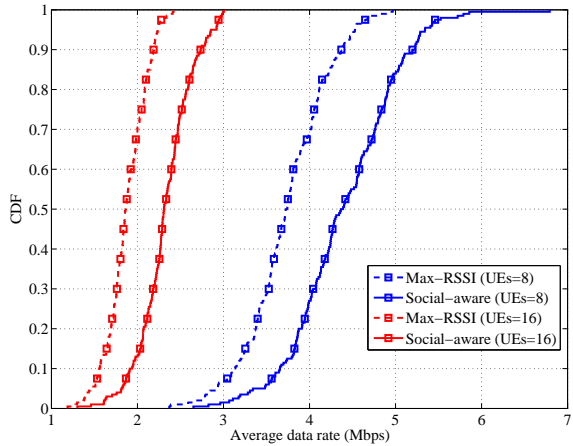}
	\vspace{-0.8cm}
	\caption{Cumulative density function of UE's rate for fixed SCBSs ($N=8$), UEs per SCBS ($M = 8$) and ($M = 16$) under the considered approaches.}
	\label{fig:cdf-ue}
\end{figure}%
\vspace{-0.3cm}
\subsection{Cell Edge Performance for Fixed Number of SCBS}
Fig. \ref{fig:cdf-ue} shows the cumulative density function of UE's rate for $N=8$ SCBSs and different sets of UEs per SCBS i.e., $M=8$ and $M=16$. In order to examine the gains in the UE's rate with different density of UEs, we analyze different percentiles of user throughput in terms of achievable data rate. Fig. \ref{fig:cdf-ue} shows that there is an increase in the data rate relative to the social unaware association approach. Table \ref{tab:cdf-ue} shows the percentiles of the UE throughput for different for different density of UEs and fixed SCBSs. In Fig. \ref{fig:cdf-ue}, we can see that, for $M=8$ UEs per SCBS, our proposed social-aware user association outperforms the social-unaware user association (max-RSSI) approach by up to $26\%$, $24\%$ and $31\%$ for $5$-th, $50$-th and $95$-th percentiles of user throughput, respectively. Furthermore, in case of $M=16$ UEs per SCBS, the social-aware user association approach shows significant gains compared to social-unaware approach (max-RSSI) in terms of data rate up to $19\%$, $18\%$ and $18\%$ for $5$-th, $50$-th and $95$-th percentiles of user throughput, respectively.

\begin{table}[!t]
	\caption{Percentiles of UE throughputs}
	\centering
	\begin{tabular}{| p{5.0cm} | p{0.7cm} | p{0.7cm} | p{0.7cm} |}
		\hline
		\textbf{Percentiles of UE throughputs} & \textbf{5-th} & \textbf{50-th} & \textbf{95-th} \\ \hline
		Average data rate gain as compare to max-RSSI approach ($M=8$)  & $26\%$  & $24\%$  & $31\%$  \\ \hline
		Average data rate gain as compare to max-RSSI approach ($M=16$) & $19\%$  & $18\%$  & $18\%$  \\ \hline
	\end{tabular}
	\label{tab:cdf-ue}
\end{table}

\begin{figure}[t]%
\centering%
\includegraphics[width = \columnwidth]{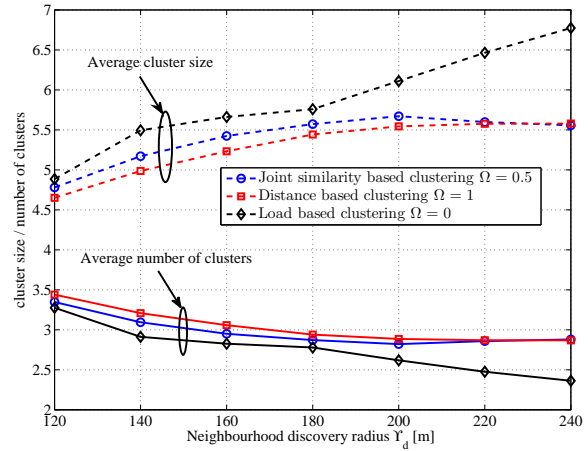}
\vspace{-0.4cm}
\caption{Comparison of average number of clusters and average cluster size with different similarities for the fixed number of UEs $M=20$ per SCBS, and $N = 16$ SCBSs.}
\label{fig:cluster-size-scbs}
\end{figure}%

\begin{figure}[t]%
\centering%
\includegraphics[width = \columnwidth]{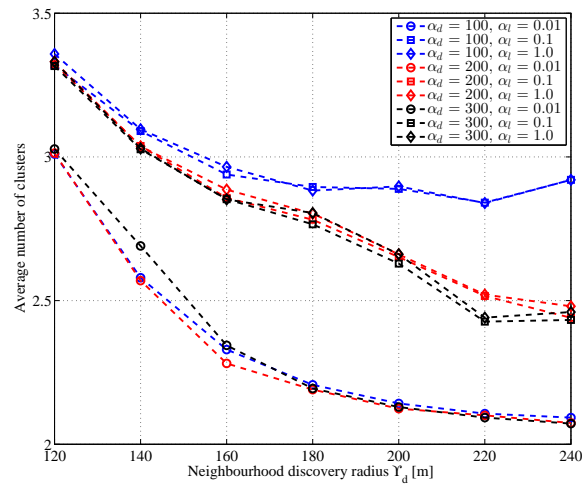}
\vspace{-0.4cm}
\caption{Average number of clusters with different $\alpha_d$ and $\alpha_l$ values for fixed number of UEs $M=20$ per SCBS, and $N=16$ SCBSs.}
\label{fig:avg-alpha-impact}
\end{figure}%
\vspace{-0.3cm}
\subsection{Impact of Similarity-based Clustering}
In Fig. \ref{fig:cluster-size-scbs}, we present the average number of clusters and the average cluster sizes of SCBSs for various approaches. We fix the number of UEs $M=20$ per SCBS and SCBSs $N=16$, with the various neighborhood discovery range $\Upsilon_d$ from $120$m to $240$m. Fig. \ref{fig:cluster-size-scbs} demonstrates the impact of distance, load, and joint similarity on the coordination of SCBSs to form clusters as per in (\ref{eq:gaussiancluster}). For the joint similarity, $\Omega$ is set to $0.5$. As per (\ref{eq:distanceSim})-(\ref{eq:gaussiancluster}), it can be shown that as the distance increases, all edges have non-zero weight between SCBSs increases. Therefore, clustering based on the distance similarity allows to group more SCBSs together yielding less average number of clusters with larger average cluster size. The increase of cluster size directly influences on the cluster load, while clustering based on the joint similarity which takes into account distance and load similarities to form clusters.
\begin{figure}[t]%
\centering%
\includegraphics[width = \columnwidth]{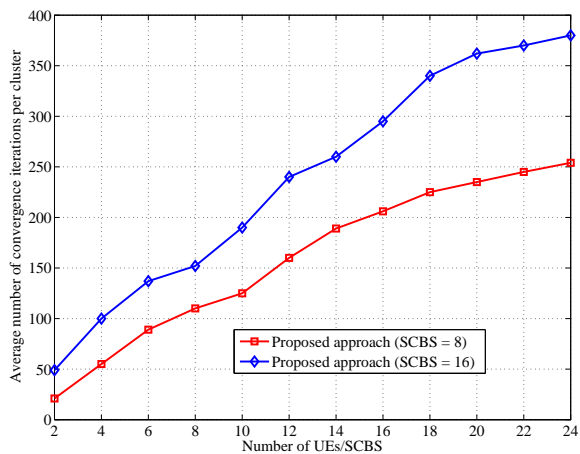}
\vspace{-0.4cm}
\caption{Average number of iterations as a function of the number of the SCBSs N and fixed number of UEs per SCBS (M=10), under the proposed approach.}
\label{fig:avg-iterations-ue}
\end{figure}%

Fig. \ref{fig:avg-alpha-impact} shows the effect of $\sigma_d$ and $\sigma_l$ on SCBS clustering. For this result, we use joint similarity based clustering such that $\Omega$ is set to $0.5$. The number of UEs per SCBS $M=20$ and number of SCBSs $N=16$ are fixed with the variation in the neighborhood discovery $\Upsilon_d$ ranging from $120m$ to $240m$. It can be shown that, by varying $\alpha_d$ the number of clusters decreases for a fixed value of $\alpha_l$ when the neighborhood discovery radius is under $160m$. It is worth mentioning that the range of the Gaussian distance similarity for any two connected SCBS is $[e^{-\Upsilon_d / 2\sigma_d^2},1]$. Thus, as the distance similarity increases, SCBSs come closer and more likely to cooperate. For a fixed value of $\alpha_d$ and varying $\alpha_l$ it can be observed that, the average number of clusters increases as $\alpha_l$ increases within the load dissimilarity range given as $[1 , e^{1/2\sigma_l^2}]$.

Fig. \ref{fig:avg-iterations-ue} shows the average number of iterations required to reach a stable matching as a function of a fixed number SCBSs $N=8$ and $N=16$ and varying number of UEs per SCBS $M$ . The average number of iterations are calculated over all the SCBSs clusters, which are the average number of iterations required per cluster to get converge. In the figure, we can observe that, as the number of UEs increases, the average number of iterations increases due to the increase in the number of players in the system. From Fig. \ref{fig:avg-iterations-ue} we can also observe that, the proposed social-aware approach requires reasonable number of iterations for convergence.
\section{Conclusions}
\label{sec:7}
In this paper, we proposed a novel, social network-aware approach for user association in D2D underlaid small cell base stations. We formulated the problem as a matching game with externalities in which the goal of each cluster of SCBSs is to maximize the social welfare which captures the data rates and peer effect due to social ties among nodes. A dynamic clustering approach is introduced to cluster base stations based on their distance and load similarities. In the proposed matching game, each UEs and SNs build their preferences and self-organize in their respective cluster as to choose their own utilities and achieve two-sided pairwise stable matching. To solve the game, we proposed social network-aware algorithm, in which UEs and SNs reach a stable matching in a reasonable number of simulation iterations. Simulations results have shown that the proposed social network-aware approach provides considerable gains in terms of increased data rates with respect to a classical social-unaware approaches.

\end{document}